\documentclass[reprint,
groupedaddress,
bibnotes,
 amsmath,amssymb,
 aps,
prb,
]{revtex4-1}

\usepackage{amsfonts, mathrsfs, graphicx}
\usepackage{bbm, bm}
\usepackage{dcolumn}

\usepackage{umoline}
\usepackage[usenames,svgnames]{xcolor}

\usepackage{hyperref}
\hypersetup{
     colorlinks=true,           
     linkcolor=Crimson,          
     citecolor=blue,            
     filecolor=blue,            
     urlcolor=cyan,             
}

\newtheorem{theorem}{Theorem}

\newtheorem{lemma}[theorem]{Lemma}
\newtheorem{lemma*}{Lemma}

\newtheorem{corollary}[theorem]{Corollary}

\newtheorem{problem}{Problem}

\newcommand{\qedsymb}{\hfill{\rule{2mm}{2mm}}}  
\newenvironment{proof}[1][]{\begin{trivlist}  
  \item[\hspace{\labelsep}{\it\noindent Proof#1:\/}]}
  {\qedsymb\end{trivlist}}

\newcommand{\Eq}[1]{Eq.~(\ref{#1})}
\newcommand{\Fig}[1]{Fig.~\ref{#1}}

\newcommand{\Lem}[1]{Lemma~\ref{#1}}
\newcommand{\Cor}[1]{Corollary~\ref{#1}}
\newcommand{\Sec}[1]{Sec.~\ref{#1}}
\newcommand{\Ref}[1]{Ref.~\onlinecite{#1}}
\newcommand{\Refs}[1]{Refs.~\onlinecite{#1}}

\newcommand{\cc}[1]{~\cite{#1}}

\newcommand{\ket}[1]{|#1\rangle}
\newcommand{\bra}[1]{\langle#1|}
\newcommand{\av}[1]{\langle#1\rangle}

\newcommand{\norm}[1]{\|#1\|}

\newcommand{\EqDef}{:=}
\newcommand{\Tr}{\mathop{\rm Tr}\nolimits}
\newcommand{\Sp}{\mathop{\rm span}\nolimits}
\newcommand{\supp}{\mathop{\rm supp}\nolimits}
\newcommand{\Id}{\ensuremath{\mathbbm{1}}}

\newcommand{\orderof}[1]{\mathcal{O}(#1)}

\newcommand{\eps}{\epsilon}

\newcommand{\gs}{\Omega}
\newcommand{\peps}{\Omega_p}

\newcommand{\DL}{\mathop{\rm DL}\nolimits}

\begin{document}

\title{How local is the information in tensor networks of matrix
  product states and entangled pairs states}

\author{Anurag Anshu}
\affiliation{Centre for Quantum Technologies,
  National University of Singapore, Singapore}
\author{Itai Arad}
\email{arad.itai@fastmail.com}
\affiliation{Centre for Quantum Technologies, National University 
  of Singapore, Singapore}
\author{Aditya Jain}
\affiliation{Center for Computational Natural Sciences and
  Bioinformatics, International Institute of Information
  Technology-Hyderabad, India}

\date{\today} 

\begin{abstract}
  Two dimensional tensor networks such as projected entangled pairs
  states (PEPS) are generally hard to contract. This is arguably the
  main reason why variational tensor network methods in 2D are still
  not as successful as in 1D. However, this is not necessarily the
  case if the tensor network represents a gapped ground state of a
  local Hamiltonian; such states are subject to many constraints and
  contain much more structure. In this paper we introduce an
  approach for approximating the expectation value of a local
  observable in ground states of local Hamiltonians that are
  represented by PEPS tensor-networks. Instead of contracting the
  full tensor-network, we try to estimate the expectation value
  using only a local patch of the tensor-network around the
  observable. Surprisingly, we demonstrate that this is often easier
  to do when the system is frustrated. In such case, the spanning
  vectors of the local patch are subject to non-trivial constraints
  that can be utilized via a semi-definite program to calculate
  rigorous lower- and upper-bounds on the expectation value. We test
  our approach in 1D systems, where we show how the expectation
  value can be calculated up to at least 3 or 4 digits of precision,
  even when the patch radius is smaller than the correlation length.
  
\end{abstract}
\maketitle

\section{Introduction}

Variational tensor-network methods\cc{ref:Orus2014-TN-Intro} provide
a promising way for understanding the low-temperature physics of
many-body condensed matter systems. In particular, they seem
suitable for studying the ground states of highly frustrated
systems, where the sign problem limits many of the quantum Monte
Carlo approaches. The best-known and by far the most successful
tensor-network method is the Density Matrix Renormalization Group
(DMRG) algorithm\cc{ref:White1992-DMRG1, ref:White1993-DMRG2}. It
has revolutionized our ability to numerically probe strongly
correlated systems in 1D, often providing results up to
machine precision. Today, it is generally believed to be the optimal
method for finding ground states of 1D lattice models.

DMRG can be viewed as a variational algorithm for minimizing the
energy of the system over the manifold of Matrix Product States
(MPS)\cc{ref:OR1995-DMRG-as-MPS1, ref:RO1997-DMRG-as-MPS2}. These
are special types of tensor-network with a linear, 1D structure. In
2D and beyond, several generalizations of this approach are
possible. Arguably, the most natural generalization is Projected
Entangled Pairs State (PEPS) tensor network\cc{ref:VC2004-PEPS, 
ref:VMC2008-PEPS-rev}, which was introduced by Verstraete and Cirac
in 2004\cc{ref:VC2004-PEPS}, but was also used earlier under
different names such as ``vertex matrix product ansatz'' in
\Ref{ref:SM1998-DMRG-CFT}, ``tensors product form ansatz'' (TPFA) in
\Ref{ref:HOK1999:NRG2D}, and ``tensor product state'' (TPS) in
\Ref{ref:NHYOMAG2001-TPS}. It is the main tensor-network that we
consider in this paper. Other constructions include Tree Tensor
Networks\cc{ref:SDV2006-TreeTN}, Multi-scale Entanglement
Renormalization ansatz (MERA)\cc{ref:Vidal2008-MERA}, String bond
states\cc{ref:SWFC2008-String-States}, and the recently introduced
Projected Entangled Simplex States (PESS)\cc{ref:XCYFNX2014-PESS},
to name a few.  These tensor-networks have proven a vital
theoretical tool for understanding the physics of 2D lattice systems
and in particular their entanglement structure. However, as a
numerical method for studying highly frustrated 2D quantum systems,
they still face substantial challenges which limit their
applicability. In most cases, the best results are still obtained
either by DMRG, in which a 1D MPS wraps around the 2D surface, or by
quantum Monte Carlo methods.  

There are several reasons for this qualitative difference between 1D
and 2D systems. The most important one is the computational cost of
contracting 2D tensor networks. While this cost scales
linearly in the 1D case, it is exponential for 2D and above.
Formally, this is reflected in the fact that contracting a PEPS is
\#P-hard\cc{ref:SWVC2007-PEPS-Complexity}, which is at least
NP-hard. To overcome this exponential hurdle, many approximation
schemes have been devised. For example, in the original PEPS paper,
the network is contracted column by column from left to right, by
treating the tensors of a column as matrix product operators (MPO)
that act on a MPS that represents the contracted part of the
network. Throughout the contraction, the bond dimension of the MPS
is truncated to some prescribed $D'$, which introduces some errors
(for details, see \Refs{ref:VC2004-PEPS, ref:VMC2008-PEPS-rev}).
Other approximate methods include, for example, the Corner Transfer
Matrix (CTM)\cc{ref:NO1996-CTM0, ref:OV2009-CTM1, ref:Orus2012-CTM2,
ref:VMVH2015-CTM3}, coarse-graining by tensor
Renormalization\cc{ref:LN2007-TRG, ref:XJCWX2009-SRG1,
ref:ZXCWCX2010-SRG2, ref:XCQZYX2012-HOTRG}, the single layer
method\cc{ref:PWV2011-Single-Layer} and other variants.

While all of the above methods are physically motivated, none of
them are rigorous. And to some extent they all produce uncontrolled
approximations, even when dealing with the ground state itself.
Moreover, while their computational cost is now polynomial in the
bond dimension and the particle number, it still scales badly, which
limits their practical use to small systems/resolutions (state of
the art now days is around $15\times 15$ sites with $D=6$\
\cc{ref:LCB2014-Cluster2}). This has led some researchers to develop
the popular simulation framework, known as the ``simple update''
method, in which one completely abandons the contraction of the 2D
network during the variational
procedure\cc{ref:JWX2008-Simple-Update}, essentially approximating
the environment of a local tensor by a product state. While this
allows for much higher bond dimensions and is often successful for
translationally invariant systems, it may produce poor results for
systems that approach critically; see, for example, the analysis in
\Refs{ref:LCB2014-Cluster1, ref:LCB2014-Cluster2}.

In this paper we introduce a new approach for approximating the
expectation values of local observables in a 2D PEPS tensor-network.
Instead of contracting the \emph{full} tensor network (or
approximating such full contraction), we aim at approximating the
expectation value using only the information inside a \emph{local}
patch of the tensor-network around the local observable. Clearly,
this approach will fail for general PEPS since they often contain
highly non-local correlations. However, as we shall demonstrate,
such an approach can produce non-trivial results when applied to
PEPS that are ground states of gapped local Hamiltonains. Indeed, it
is known that such states exhibit strong properties of locality,
such as exponential decay of correlations\cc{ref:Hastings2004-EXP,
ref:HK2006-EXP, ref:Nachtergaele2006-LR} and local
reversibility\cc{ref:KAAV2015-LR}, and are therefore subject to many
constraints to which arbitrary PEPS are not. Moreover, we have the
local Hamiltonian at our hands, which can be used to evaluate the local
expectation value. Our goal in this paper is to show that for these
states, very good approximations for the local expectation values can
be derived only from a local patch.

We identify two possible mechanisms for such approximations, which
form the basis for two numerical algorithms. In both algorithms, the
entire calculation is local and can therefore be done, in principle,
efficiently.  Moreover, both algorithms provide \emph{rigorous}
upper- and lower- bounds on the expectation value. While we usually
cannot give rigorous bound on the distance between these bounds, we
demonstrate numerically that this distance -- and hence the error in
our approximation -- can be surprisingly small.

The first algorithm, which we call the `basic algorithm', is
expected to give good results in the case of frustration-free gapped
systems. The second one, which we call the `commutator gauge
optimization' (CGO) algorithm, works only for frustrated systems by
utilizing the additional constraints in these systems. It does not
rely directly on the existence of a gap, and may work even when
considering patches of the PEPS that are much smaller than the
correlation length. We numerically demonstrate both algorithms using
MPS in 1D systems. Regarding the title of this
paper, our findings suggest that when the PEPS/MPS tensor-network
represents a frustrated ground state, the information about the
local expectation value is largely encoded locally in the
neighborhood of the observable. 

We would like to stress from the start that the algorithms we
present are \emph{not} practical, and can only be used in 1D in a
reasonable time. Their main goal, which is the main goal of this
paper, is to illustrate the \emph{locality of information} in
MPS/PEPS representations of gapped ground states. Nevertheless, we
strongly believe that the mechanisms behind these algorithms can be
further exploited and turned into practical heuristic algorithms. We
leave this interesting research direction for further work.

The organization of this paper is as follows. In \Sec{sec:setup} we
formally define the problem we wish to solve and the assumptions we
are using. In \Sec{sec:basic} we introduce the basic algorithm to
solve the problem, and in \Sec{sec:CGO} we introduce the CGO
algorithm. In \Sec{sec:numerics} we present the results of the 1D
numerical tests, and in \Sec{sec:discussion} we discuss the results
and offer our conclusions.

\section{Statement of the problem}
\label{sec:setup}

Our construction involves PEPS and MPS tensor networks. For the
exact definition and review of these tensor networks (and tensor
networks in general), we refer the reader to
\Refs{ref:Orus2014-TN-Intro, ref:VMC2008-PEPS-rev}.

We are given a local Hamiltonian $H=\sum_i h_i$ that is defined on a
system of $N\times N$ spins of local dimension $d$ that sit on a 2D
rectangular lattice, which can be either open or with periodic
boundary conditions. The local terms $h_i$ have $\orderof{1}$
bounded norms, and are assumed to be working on nearest-neighbors
only. We further assume that $H$ is gapped, with a unique ground
state $\ket{\gs}$ that is \emph{exactly} described by a PEPS with
bond dimension $D=\orderof{1}$.  While in practice this is rarely
the case, and $D$ may have to be exponentially large in order for
the PEPS to exactly describe the ground state, it is expected that a
$D=\orderof{1}$ bond dimension should give very good approximations
to a gapped ground state. For the sake of clarity we first assume
that the $D=\orderof{1}$ description is exact and return to discuss
this assumption in \Sec{sec:constant-D}.

Given a local observable $B$ that acts on, say, 2 neighboring spins
on the lattice, our task is to approximate $\av{B}\EqDef
\bra{\gs}B\ket{\gs}$ using only a local patch of the PEPS around
$B$. Specifically, the local patch is defined by a ball $L$ of
radius $\ell$ around $B$. $L$ contains all the sites on the lattice
that can be connected to the support of $B$ using at most $\ell$
steps on the lattice. We let $L^c$ denote the complement region of
$L$ that contains all the spins outside of $L$. An example of the
local ball $L$ for $\ell=3$ is given in \Fig{fig:Lregion}. 
\begin{figure}
  \includegraphics[scale=1.0]{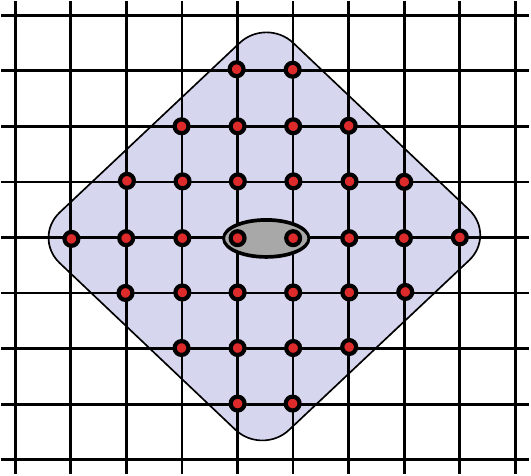}
  \caption{\label{fig:Lregion} An example of a ball $L$ of
  radius $\ell=3$ around a 2-local operator. The region $L$ consists
  of the $32$ spins that are marked by red dots.}
\end{figure}

Without loss of generality, we will assume that $B$ is normalized so
that $\norm{B}=1$; all of our results can be applied to the general
case with a simple rescaling.

Consider then the PEPS representation of $\ket{\gs}$, and let
$\alpha=(\alpha_1, \alpha_2, \ldots,)$ be the set of virtual indices
that connect the spins in $L$ to those outside of it. As illustrated in
\Fig{fig:Ialpha}, it defines the following decomposition of
$\ket{\gs}$: 
\begin{align}
  \ket{\gs} = \sum_\alpha \ket{O_\alpha}\otimes \ket{I_\alpha} \,.
\end{align}
Above, $\{\ket{O_\alpha}\}$, $\{\ket{I_\alpha}\}$ are states that are
defined \emph{by the PEPS tensor-network} outside and inside $L$
respectively.  The sets of vectors
$\{\ket{O_\alpha}\}$, $\{\ket{I_\alpha}\}$ are not necessarily
orthogonal between themselves, nor are they normalized, but we
assume that both $\{\ket{I_\alpha}\}$ and $\{\ket{O_\alpha}\}$ are
linearly independent among themselves. This is a stronger condition
than the injectivity condition for the PEPS, in which only
$\{\ket{I_\alpha}\}$ have to be linearly
independent\cc{ref:PVWC2008-UniquePEPS}, but as injectivity, we
expect it to hold for generic PEPS. We define $V_L\EqDef
\Sp\{\ket{I_\alpha}\}$ and let $\rho_L$ denote the reduced density
matrix of $\ket{\gs}$ on $L$, which lives in the subspace $V_L$.
Note that from our assumptions on $\{\ket{I_\alpha}\}$ and
$\{\ket{O_\alpha}\}$ it follows that $V_L = \mathrm{Im} \rho_L$. We
let $q\EqDef\dim V_L$, which is equal to the size of the range in
which the composite index $\alpha=(\alpha_1, \alpha_2, \ldots)$
runs. Notice that $q=D^{\orderof{|\partial L|}}$, reflecting the
fact that $\ket{\gs}$ satisfies an area law.
\begin{figure}
  \includegraphics[scale=0.9 ]{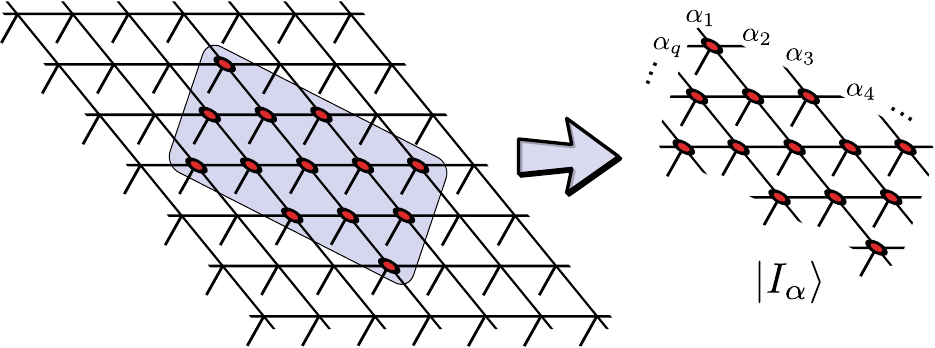}
    \caption{\label{fig:Ialpha} Illustration of the decomposition
    $\ket{\gs}=\sum_\alpha \ket{O_\alpha}\otimes\ket{I_\alpha}$,
    defined by a ball $L$ and the underlying PEPS. The boundary of
    $L$ cuts the PEPS virtual links $\alpha=(\alpha_1, \alpha_2, 
    \ldots, \alpha_q)$ and defines the outer states $\{
    \ket{O_\alpha}\}$ and the inner states $\{\ket{I_\alpha}\}$,
    which are described on the right part of the figure. The inner 
    vectors $\{\ket{I_\alpha}\}$ and the subspace $V_L$ that they span
    can be efficiently calculated.}
\end{figure}

Clearly, if we knew $\rho_L$ we could calculate $\av{B}$ from
$\av{B}=\bra{\gs}B\ket{\gs} = \Tr_L(\rho_LB)$. However, directly
estimating $\rho_L$ involves the contraction of the full network,
which is exactly what we are trying to avoid. Instead, since the
radius $\ell$ is assumed to be small, we can easily calculate
$V_L\EqDef \Sp\{\ket{I_\alpha}\}$.  Can we estimate $\av{B}$ using
only $V_L$ and the fact that $\ket{\gs}$ is the ground state of $H$?
Formally, we ask:
\begin{problem}
  Given a local observable $B$, a ball $L$ of radius $\ell$ around
  it, and the corresponding subspace $V_L$, find a range
  $[b_{min},b_{max}]$, as narrow as possible, such that
  $\av{B}\in[b_{min},b_{max}]$.
\end{problem}

We offer two algorithms to tackle this problem. The first, which we
call the `basic algorithm', is expected to work well when the system
is gapped, and is suitable mostly for frustration-free systems. The
second, which we call `the commutators gauge optimization' algorithm,
is much more powerful and applies only for frustrated systems. It
relies on the assumption that $\ket{\gs}$ with $D=\orderof{1}$ is
described by a PEPS and therefore satisfies an area law (this, in
turn, can also be attributed to the existence of a gap). We begin
with the basic algorithm.

\section{The basic algorithm}
\label{sec:basic}

Let $P_V$ be the projector into the subspace $V_L$, and define
$B_L\EqDef P_V B P_V$. Then as $\rho_L = P_V\rho_LP_V$, we get from
the cyclicity of the trace that $\av{B}=\Tr(B\rho_L) =
\Tr(BP_V\rho_LP_V) = \Tr(P_V B P_V \rho_L) = \Tr(B_L \rho_L)$.
Equivalently, $\av{B} = \bra{\gs} B_L\ket{\gs}$. Therefore, if $b_1
\le b_2 \le \ldots \le b_q$ are the eigenvalues of $B_L$ in the
subspace $V_L$, then $\av{B} \in [b_1, b_q]$. The basic algorithm
then simply calculates the largest and the smallest eigenvalues of
$B_L$ in the $V_L$ subspace and use these as lower- and upper-
bounds to $\av{B}$.

Intuitively, when the system is gapped, we expect the range $[b_1,
b_q]$ to narrow down as $\ell\to\infty$ for the following reason. 
Every $\alpha=(\alpha_1, \ldots, \alpha_q)$ can be seen as a
specific boundary condition for a restricted system on the ball $L$,
and as the system is gapped, we expect the effect of these boundary
conditions to be negligible. Hence, we expect that for every
$\ket{I_\alpha}$, we have $\bra{I_\alpha}
B\ket{I_\alpha}/\norm{I_\alpha}^2 \to \av{B}$ as $\ell\to\infty$,
and similarly for any linear combination of the $\ket{I_\alpha}$
vectors. In particular this means that the eigenvalues of $B_L$ are
expected to converge to $\av{B}$.  The following lemma shows that
this is indeed the case when the eigenvalues of $\rho_L$ are not too
small.
\begin{lemma}
\label{lem:expdec}
  Let $\lambda^2_{min}$ be the minimal eigenvalue of $\rho_L$, and
  let $b_1, \ldots, b_q$ be the eigenvalues of $B_L = P_V B P_V$ in
  the subspace $V_L$. Then for every $\beta=1, \ldots, q$,
  \begin{align}
    |b_\alpha -\av{B}| \le
      \frac{\norm{B_L}}{\lambda_{min}^2}e^{-\ell/\xi} \,,
  \end{align}
  where $\xi>0$ is the correlation length of $\ket{\gs}$.
\end{lemma}
\begin{proof}
  The proof is a simple application of the basic idea in the
  Choi-Jamio\l kowski isomorphism, together with the exponential
  decay of correlation property of gapped ground
  states\cc{ref:Hastings2004-EXP}. Let $\ket{b_1},\ket{b_2},\ldots
  \ket{b_q}$ be the eigenvectors of $B_L$ in $V_L$ that correspond
  to the eigenvalues $b_1, b_2, \ldots, b_q$. They constitute an
  orthonormal basis of $V_L$. Let $\ket{\gs} = \sum_\alpha
  \lambda_\alpha \ket{\hat{O}_\alpha}\otimes \ket{\hat{I}_\alpha}$
  be the Schmidt decomposition of $\ket{\gs}$ with respect to $L$
  and $L^c$ such that $\rho_L = \sum_\alpha \lambda^2_\alpha
  \ket{\hat{I}_\alpha}\bra{\hat{I}_\alpha}$. Both
  $\{\ket{b_\beta}\}$ and $\{\ket{\hat{I}_\alpha}\}$ are orthonormal
  bases of $V_L$, and therefore they are connected by a unitary
  transformation: $\ket{b_\beta} = \sum_\alpha
  U_{\beta\alpha}\ket{\hat{I}_\alpha}$. We use it to define a set of
  (un-normalized) states on the spins of $L^c$:
  \begin{align}
  \label{eq:c-beta}
    \ket{c_\beta} 
      \EqDef \sum_\gamma \frac{1}{\lambda_\gamma}U^*_{\beta\gamma}
        \ket{\hat{O}_\gamma} \,.
  \end{align}
  In addition, we define the operator $C_\beta \EqDef
  \ket{c_\beta}\bra{c_\beta}\otimes \Id_L$, where $\Id_L$ is the
  identity operator on the spins of $L$. Using
  these definitions, the following identities are easy to verify:
  \begin{align}
  \label{eq:prop1}
    C_\beta \ket{\gs} &= \ket{c_\beta}\otimes\ket{b_\beta} \,, \\
  \label{eq:prop2}
    \norm{C_\beta} &= \norm{c_\beta}^2 
      = \sum_\gamma \frac{1}{\lambda^2_\gamma} 
        |U_{\beta\gamma}|^2 \,, \\
    C_\beta^2 &= \norm{c_\beta}^2 C_\beta \,.
  \label{eq:prop3}
  \end{align}
  
  Next, by using the exponential decay of
  correlations\cc{ref:Hastings2004-EXP} and the uniqueness of the
  ground state, we deduce that 
  \begin{align}
  \label{eq:decay}
    &\big|\bra{\gs}C_\beta B\ket{\gs}
      - \bra{\gs}C_\beta\ket{\gs}
          \cdot \bra{\gs}B\ket{\gs}\big| \\
       &\ \le \norm{B}\cdot\norm{C_\beta}\cdot e^{-\ell/\xi} \,.
       \nonumber
  \end{align}
  Let us analyze the above inequality term by term. First, by
  \Eq{eq:prop3} and the fact that $C_\beta$ and $B$ commute,
  $\bra{\gs}C_\beta B\ket{\gs} = \frac{1}{\norm{c_\beta}^2}
  \bra{\gs}C_\beta B C_\beta\ket{\gs}$. But then by \Eq{eq:prop1},
  this is equal to $\bra{b_\beta}B\ket{b_\beta} 
  =\bra{b_\beta}B_L\ket{b_\beta} = b_\beta$. A
  similar argument shows that $\bra{\gs}C_\beta\ket{\gs}=1$.
  Substituting this into the LHS of \eqref{eq:decay}, we get
  $|b_\beta-\av{B}| \le \norm{B}\cdot \norm{C_\beta} \cdot
  e^{-\ell/\xi}$. Finally, using \Eq{eq:prop2} and the fact that
  $\sum_\gamma |U_{\beta\gamma}|^2=1$ (since $U_{\beta\gamma}$ are
  the entries of a unitary matrix), proves the lemma.
\end{proof}

\Lem{lem:expdec} shows that if the smallest eigenvalue of $\rho_L$
is lowerbounded by $e^{-c\ell/\xi}$ for some $c<1$, the range $[b_1,
b_q]$ should converge to the point $\av{B}$ exponentially fast in
$\ell$. This is certainly expected in most 1D ground states that are
approximated by an MPS of a constant bond dimension $D$. However, it
can hardly happen in 2D, where even if the state is approximated by
a PEPS with $D=\orderof{1}$, the dimension of $V_L$ grows at least
as $D^{\orderof{\ell}}$, and consequently, the smallest eigenvalue
of $\rho_L$ is expected to fall off exponentially like
$D^{-\orderof{\ell}}$. In 3D things are even worse, since $\dim V_L$
grows like $D^{\orderof{\ell^2}}$.

When the underlying Hamiltonian is frustration-free, $V_L$ is a
subspace of the groundspace of $H_L$, the part of $H$ that is
supported on the spins of $L$. In such case, it is easy to see that
if the Hamiltonian satisfies the local topological order
(LTQO) property\cc{ref:MZ2013-LTQO1} (see also \Ref{ref:CMPS2013-LTQO2}),
 then necessarily $|b_\alpha-\av{B}|$ decays as a function
of $\ell$ (the rate of the decay depends on the particular
definition of LTQO). Related to that, \Lem{lem:expdec} has much in
common with Theorem~11 of \Ref{ref:CMPS2013-LTQO2}, which shows that
parent Hamiltonians of translationally invariant, injective MPS
satisfy LTQO.

We conclude this section with a somewhat stronger result than
\Lem{lem:expdec} that holds for frustration-free systems:
\begin{lemma}
\label{lem:FF-AGSP}
  When $\ket{\gs}$ is the unique ground state of a frustration-free
  Hamiltonian with an $\orderof{1}$ spectral gap, then 
  \begin{align}
  \label{eq:AGSP-ineq}
    B_L\ket{\gs} = \av{B}\ket{\gs} + \ket{\delta} \,,
  \end{align}
  where $\norm{\ket{\delta}}\le e^{-\orderof{\ell}}$. In other
  words, $\ket{\gs}$ is an approximate eigenvector of $B_L$.
\end{lemma}
The proof of this lemma is given in the Appendix. Formally, it is at
least as strong as \Lem{lem:expdec} because it can be used to derive
\Lem{lem:expdec} when the system is frustration-free. Indeed, using
the notation of the proof of \Lem{lem:expdec}, the lemma can be
proved by multiplying inequality~\eqref{eq:AGSP-ineq} by $C_\beta$
and then by $\bra{c_\beta}\otimes\bra{\beta_\beta}$. But it also
seems stronger since it tells us something about the eigenvalues of
$B_L$ even when $\lambda_{min}^2$ is much smaller then
$e^{-\ell/\xi}$. To see this, note that \eqref{eq:AGSP-ineq} implies
that 
\begin{align*}
  \norm{(B_L-\av{B}\Id)^2\ket{\gs}}^2
    = \norm{\delta}^2\le e^{-\orderof{\ell}}\,.
\end{align*}
If $\{\ket{b_\beta}\}$ is the eigenbasis of $B_L$ with eigenvalues
$b_1, b_2, \ldots$, then the above inequality can be written as
\begin{align}
\nonumber
  &\norm{(B_L-\av{B}\Id)^2\ket{\gs}}^2 
    = \Tr_L\big[\rho_L(B_L-\av{B}\Id)^2\big] \\
  &= \sum_\beta \bra{b_\beta}\rho_L\ket{b_\beta}(b_\beta-\av{B})^2
   \le  e^{-\orderof{\ell}} \,.
\end{align}
As $\{\bra{b_\beta}\rho_L\ket{b_\beta}\}$ is the probability
distribution that corresponds to the measurement of $B_L$, we can
use Markov inequality to deduce that if we measure $B_L$ on
$\ket{\gs}$, then with probability of at least
$1-e^{-\orderof{\ell}}$ we will obtain an eigenvalue of $B_L$ that
is $e^{-\orderof{\ell}}$ close to $\av{B}$. This holds regardless of
the minimal eigenvalue of $\rho_L$. Therefore in cases where the
weight of all the eigenvalues of $B_L$ in $\rho_L$ is of the same
order (which can be much much smaller than $e^{-\orderof{\ell}}$ in
3D), then an exponentially large fraction of them must be
exponentially close to $\av{B}$. This does not prove that the range
$[b_{min}, b_{max}]$ rapidly shrinks, but it supports the intuition
that this should generally be the case.

We conclude this section noting that for frustrated systems a
similar result to \Lem{lem:FF-AGSP} can be proved using the same
techniques that are used to prove the exponential decay of
correlations in the frustrated case\cc{ref:Hastings2004-EXP}: a
combination of Lieb-Robinson bounds and a suitable filtering
function. In fact, a very similar result was already proven by
Hastings in \Ref{ref:Hastings2006-GappedH} for the special case when
$B$ is a local Hamiltonian term. However, instead of pursuing that
direction, we turn to a different algorithm, which turns out to be
much more powerful for our problem.

\section{The commutator gauge optimization}
\label{sec:CGO}

In this section we introduce the Commutator Gauge Optimization
algorithm (CGO for short), which is applicable for frustrated
systems. As we shall see, it can be viewed as pair of primal-dual
SDP optimization problem. We start with the formulation of the
primal optimization problems.

%
%
\subsection{Primal problem}

Our starting point is the simple observation of the basic algorithm,
that the expectation value $\av{B}$ must be inside eigenvalues range
$[b_{min}, b_{max}]$ of the operator $B_L\EqDef P_V B P_V$. The additional
idea is that this must hold \emph{for any other operator $K$ for
which $\av{K}=\av{B}$}. Therefore, if we can find an operator $K$
such that $\av{B}=\av{K}$, then necessarily $\av{B}$ must also be
inside the range $[k_{min}, k_{max}]$ of the minimal and maximal
eigenvalues of $P_V K P_V$. By optimizing over a subset of such
operators, we may significantly narrow down the range in which
$\av{B}$ is found. We therefore look for operators, $K_{min},
K_{max}$ on $L$ such that
\begin{align}
\label{eq:main}
  \bra{\gs}K_{min}\ket{\gs} = \bra{\gs}K_{max}\ket{\gs} 
    = \bra{\gs}B\ket{\gs} \,,
\end{align}
but at the time, $P_L K_{max} P_L$ has the smallest possible maximal
eigenvalue, and $P_L K_{min} P_L$ has the largest possible minimal
eigenvalue. 

How can we find such operators? A very simple yet
powerful trick is to look for operators of the form:
\begin{align}
  K_{min} &= B + [H, A_{min}] , &
  K_{max} &= B + [H, A_{max}] \,,
\end{align}
where $A_{min}, A_{max}$ are anti-Hermitian such that $K_{min},
K_{max}$ are Hermitian.  For any eigenstate $\ket{\eps}$ of $H$ and
for any operator $A$, it is easy to verify that
$\bra{\eps}[H,A]\ket{\eps}=0$ and therefore \Eq{eq:main} holds. In
that sense $K_{min}$ and $K_{max}$ can be viewed as a different
``gauges'' of $B$, which have the same expectation value with respect to
$\ket{\gs}$.

By taking the support of $A_{min}, A_{max}$ to be small enough, we
can guarantee that $[H, A_{min}]$ and $[H, A_{max}]$ are supported
in $L$. Formally, we partition the spins of $L$ into two disjoint
subsets, $L=\partial L \cup L_0$. Here, $\partial L$ contains the
spins in $L$ that are coupled to spins in $L^c$ via one or more
local terms in $H$, and $L_0$ contains the rest of the spins in $L$.
An illustration of this decomposition is given in
\Fig{fig:L0region}. Letting $H_L$ denote the sum of all the local
$h_i$ terms whose support is inside $L$, we conclude that when the
support of $A$ is inside $L_0$, then $[H,A] = [H_L,A]$, and this
commutator is supported inside $L$.
\begin{figure}
  \includegraphics[scale=1.0]{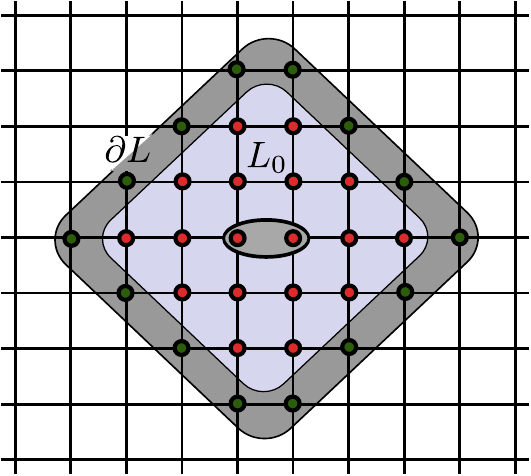}
  \caption{\label{fig:L0region} The same $\ell=3$ ball $L$ from
  \Fig{fig:Lregion}, now decomposed into the boundary region
  $\partial L$ (dark gray, green vertices) and the internal region
  $L_0$ (light gray, red vertices). }
\end{figure}

For a given Hermitian operator $O$, denote by $\lambda_{min}(O),
\lambda_{max}(O)$ its minimal/maximal eigenvalues. Then the CGO
algorithm can be summarized by the following optimization problem:
\begin{problem}[The primal CGO optimization]
\label{def:P-CGO}
  Given a local observable $B$, a ground state $\ket{\gs}$ of a
  local Hamiltonian $H$ in the form of a PEPS, and a ball $L$ of
  radius $\ell$ around $B$, calculate
  \begin{align}
  \label{eq:Pmax}
    k^{(P)}_{max} &\EqDef \min_{A} 
      \lambda_{max} \big( P_V(B + [H_L, A])P_V\big) \\
    \text{and}\nonumber \\  
  \label{eq:Pmin}
    k^{(P)}_{min} &\EqDef \max_{A} 
      \lambda_{min} \big( P_V(B + [H_L, A])P_V\big)\,.
  \end{align}  
  Above, the optimizations are over all anti-Hermitian operators $A$
  that are supported on $L_0$.
\end{problem}
Two notes are in order. First, the presence of $P_V$ in the above
equations is crucial, as it is the only place where information
about the ground state is entering the problem. Without it, values
from other eigenstates of $H$ might worsen our estimate. It is
therefore clear why some sort of an area-law is required for CGO to
work; if the support of $\rho_L$ (i.e., the subspace $V$) is the
entire available Hilbert space, there would be nothing in $V_L$ to
tell us that we are dealing with the ground state and not, say, some
other eigenstate of $H$.

Second, when $H$ is frustration free, $P_V H_L=H_LP_V =0$ and so
$P_V(B + [H_L, A])P_V = P_V B P_V$ for all $A$; the CGO algorithm then
simply becomes the basic algorithm.

Before discussing how the optimization Problem~\ref{def:P-CGO} can
be done in practice, let us introduce its dual.

%
%
\subsection{The dual problem}

To introduce the dual problem, we begin with the constraint 
$\bra{\gs}[H_L, A]\ket{\gs}=0$, which holds for every $A$ that is
supported on the spins in $L_0$.  Tracing out the spins in $L^c$,
we get
\begin{align}
  \Tr(\rho_L [H_L, A]) = 0 \quad \forall A \in \supp(L_0) \,.
\end{align}
Using the identity $\Tr(A[B,C])=\Tr(C[A,B])$, we arrive at $\Tr(A
[\rho_L, H_L]) = 0$. Writing $\Tr = \Tr_{L_0} \Tr_{\partial L}$, the
last equation can be rewritten as $\Tr_{L_0}A\big( \Tr_{\partial
L}[\rho_L, H_L]\big) = 0$. But since $A$ is an arbitrary
anti-Hermitian operator on $L_0$, it implies that the expression
it multiplies has to vanish. We therefore reach the
following corollary:
\begin{corollary}
\label{cor:rhoL-constraint}
  Let $\rho_L$ be the reduced density matrix of an eigenstate of $H$ in
  the region $L$. Then, 
  \begin{align}
    \label{eq:rhoL-constraint}
    \Tr_{\partial L} [\rho_L, H_L] = 0 \,.
  \end{align}  
\end{corollary}
The above \emph{local} identity holds for all eigenstates of $H$ and
for all regions $L$, regardless of any assumption about a gap or the
existence of a PEPS description. When $H$ is frustration free and
$\rho_L$ is the reduced density matrix of the ground state, it is
satisfied trivially since $[\rho_L, H_L]=0$. However, when the
system is frustrated, \Eq{eq:rhoL-constraint} provides many
non-trivial constraints as it holds \emph{pointwise} on the Hilbert
space of $L_0$. 

Corollary~\ref{cor:rhoL-constraint} allows us to formulate a dual
to the optimization problem given in Problem~\ref{def:P-CGO}.
Instead of optimizing over anti-Hermitian operators $A$, we may
optimize over `legal' reduced density matrices, which satisfy
\Eq{eq:rhoL-constraint}. Formally, 
\begin{problem}[The dual CGO problem]
  Given a local observable $B$, a ground state $\ket{\gs}$ of a
  local Hamiltonian $H$ in the form of a PEPS, and a ball $L$ of
  radius $\ell$ around $B$, calculate
  \begin{align}
  \label{eq:Dmax}
    k^{(D)}_{max} &\EqDef \max_{\rho_L\in S_L} \Tr(\rho_L B) \\
    \text{and}\nonumber \\  
  \label{eq:Dmin}
    k^{(D)}_{min} &\EqDef \min_{\rho_L\in S_L} \Tr(\rho_L B) \,.
  \end{align}  
  Above, $S_L$ is the set of all density matrices $\rho_L$ in $V_L$
  [i.e., $\rho_L\succeq 0$, $\Tr(\rho_L)=1$] that satisfy
  \Eq{eq:rhoL-constraint}.
\end{problem}

One can prove directly that $k^{(P)}_{max} \ge k^{(D)}_{max}$ and
that $k^{(P)}_{min} \le k^{(D)}_{min}$. We omit the proof since it
will follow naturally from viewing these two optimizations as dual
semi-definite programs (SDPs), as we shall now do.

%
%
\subsection{The commutator gauge optimization as a SDP}

The primal CGO problem and its dual can be cast as semi-definite
programs. Let us demonstrate it for the upperbound of $\av{B}$; the
lowerbound follows similarly. To write \eqref{eq:Pmax} as an SDP,
note that a number $\lambda$ is an upperbound to all eigenvalues of
an operator $O$ iff $\lambda\Id - O \succeq 0$, where $\succeq 0$
stands for non-negative matrix. Then choosing a basis $A_1, A_2,
\ldots, A_m $ for all the anti-Hermitian matrices in $L_0$, the optimization
over~\eqref{eq:Pmax} can be written as 
\begin{align}
\label{eq:SDP-P}
\text{minimize:\quad } & b \\
  \text{subject to:\quad} &b\Id 
    - \sum_{i=1}^m c_i P_V[H_L, A_i]P_V  - P_VBP_V \succeq 0
\end{align}
Above, the optimization is done over the real numbers $b, c_1, c_2,
\ldots, c_m$. The
dual SDP is then
\begin{align}
\label{eq:SDP-D}
  \text{maximize:\quad} &\Tr(\rho_L P_VBP_V) \\
  \text{subject to:\quad} &\rho_L\succeq 0, \ 
  \Tr(\rho_L\Id) = 1 , \\
   &\text{and} \quad
   \Tr(\rho_L P_V[H_L, A_i] P_V) = 0\ \ \forall i\nonumber
\end{align}
which is identical to the optimization in \eqref{eq:Dmax}. By the
weak duality of semi-definite programming, we note that the value in
\eqref{eq:SDP-D} lowerbounds the one in \eqref{eq:SDP-P}, and that
in many cases they are identical. In the next section, we present
the results of some numerical tests we performed on 1D systems to
check the performance of such SDPs.

%
%
\subsection{Working with constant bond dimensions}
\label{sec:constant-D}

Up to now, we have assumed that the ground state $\ket{\gs}$ is
given \emph{exactly} as a PEPS with constant bond dimension $D$.
However, in virtually all frustrated systems, a constant $D$ can
only give an approximation to the ground state. To distinguish the
actual PEPS approximation from the exact ground state, we will
denote it by $\ket{\peps}$. Given that $\ket{\peps}$ is only an
approximation to the ground state $\ket{\gs}$, we can only expect 
$\bra{\peps}[H,A]\ket{\peps}$ to be close to zero, but not
completely vanish. Similarly, \Eq{eq:rhoL-constraint} is expected to
be only approximately satisfied. While this may look merely as an
aesthetic defect, it raises a conceptual problem: if $D$ is
constant, the dimension of $V_L$ increases like
$D^{\orderof{\ell}}$, while the number of constraints in
\Eq{eq:rhoL-constraint} goes like $d^{\orderof{\ell^2}}$. For large
enough, yet constant $\ell$, it is expected that
\Eq{eq:rhoL-constraint} is over-determined, or simply unsatisfiable.
From an SDP point of view, we expect the constraints in the primal
and the dual problems to become linearly dependent, causing the dual
problem not to have any feasible solution and the primal problem to
yield infinities. 

A practical (though probably not optimal) solution to this is to use
only a \emph{subset} of all possible $A_i$ in the SDP procedure.
Note that this only relaxes the program in \Eq{eq:SDP-D} (and
corresponding program for minimum eigenvalue) and hence increases
the range $[k^{(D)}_{max},k^{(D)}_{min}]$.  One way to do it is to
create an increasing list of random $A_i$ matrices and feed it to
the SDP until they become linearly dependent. As long as the
matrices are not linearly dependent and are of a comparable norm, we
expect the solution of SDPs~\eqref{eq:SDP-P} and~\eqref{eq:SDP-D} with
$P_V$ taken from $\ket{\peps}$ to produce results that are close to
the exact case. The 1D numerical tests that we present in the
following section support this intuition.

\section{Numerical tests}
\label{sec:numerics}

Due to the reduced computational cost, we performed all of our tests
on 1D systems whose ground states are described by MPS.
Nevertheless, we strongly believe that our conclusions hold also for
2D systems, which we leave for future research. 

We analyzed four well-known systems, which we label System-A,
System-B, System-C, and System-D. All systems are defined over
$N=100$ spins with open boundary conditions. The underlying
Hamiltonians are nearest-neighbors with unique ground states and an
$\orderof{1}$ spectral gap and correlation lengths which were
calculated numerically.  The full details are as follows:

\paragraph*{System-A - 1D AKLT.} This is a frustration-free system
  on spin one particles ($d=3$)\cc{ref:AKLT1987, ref:AKLT1988} given
  by
  \begin{align}
    H_{sysA}\EqDef \sum_{i=1}^{N-1}\big[ 
      \frac{1}{2}\bm{S}_i\cdot\bm{S}_{i+1} 
        + \frac{1}{6}(\bm{S}_i\cdot\bm{S}_{i+1})^2
        + \frac{1}{3}\big] \,.
  \end{align}
  The ground state is known analytically, and is described by the
  AKLT valence bond state, which can be written as a MPS with bond
  dimension $D=2$. The spectral gap in the thermodynamic limit is
  $\Delta_A\simeq 0.35$\cc{ref:Knabe1988-AKLTgap}. The correlation
  length is $\xi_A\simeq 0.9$.

\paragraph*{System-B - Transverse Ising model.} This is the classical
  Ising model equipped with a transverse magnetic field in the
  $\hat{z}$ direction:
  \begin{align}
    H_{sysB}\EqDef -\frac{1}{2\sqrt{1+h^2}}\left[
      \sum_{i=1}^{N-1} S^x_iS^x_{i+1} 
        + h\sum_{i=1}^N S^z_i\right] \,.
  \end{align}
  At $h=1.0$ the model experiences a phase-transition and the gap
  closes down. We used $h=1.1$ for which the gap is $\Delta_B\simeq
  0.07$ and the correlation length is $\xi_B\simeq 8.8$.

\paragraph*{System-C - Transverse XY model.} This model resembles 
  the transverse Ising model, but with an additional
  $S^y_iS^y_{i+1}$ interaction term:
  \begin{align}
    H_{sysC}&\EqDef -\frac{1}{2\sqrt{1+h^2}}\left[
      \sum_{i=1}^{N-1} \frac{1-\alpha}{2}S^x_iS^x_{i+1}\right. \\
       &+ \left.\frac{1+\alpha}{2}S^y_iS^y_{i+1}
        + h\sum_{i=1}^N S^z_i\right] \,. \nonumber
  \end{align}
  We used $h=1.1$ and $\alpha=0.5$, for which the gap was evaluated
  numerically to be $\Delta_C\simeq 0.07$ and the correlation length
  $\xi_C\simeq 4.3$.
    
\paragraph*{System-D - XY model with random field.} Just like
  System-C, only that here the transverse field at spin $i$ is given
  by $h_i$, which is a uniformly distributed random number in the
  range $[1.05,1.15]$. 
  \begin{align}
    H_{sysC}&\EqDef -\frac{1}{2\sqrt{1+h^2}}\left[
      \sum_{i=1}^{N-1} \frac{1-\alpha}{2}S^x_iS^x_{i+1}\right. \\
       &+ \left.\frac{1+\alpha}{2}S^y_iS^y_{i+1}
        + h_i\sum_{i=1}^N S^z_i\right] \,. \nonumber
  \end{align}
  For $h=1.1,\alpha=0.5$ and $h_i\sim [1.05,1.15]$, the gap was
  evaluated numerically to be $\Delta_D\simeq 0.06$. The average
  correlation length was just like as in System-C.

\begin{figure*}[t]
  \caption{\label{fig:sysD-PxPx-converge} The progress of the 
    SDP bounds for the
    $\ell=3$ (dashed black lines) and $\ell=4$ (solid blue lines)
    lines for System-D (XY model with random transverse field) with
    the observable $B=P_xP_x$. The exact expectation value is
    $\av{B}=0.19999$. The $\ell=3$ algorithm stopped after using
    $684$ and $706$ matrices for the upperbound and lowerbounds
    respectively, yielding the bounds $[0.199968,0.200913]$. The
    $\ell=4$ algorithm stopped after using $684$ matrices in both
    upper/lower bounds, yielding the bounds $[0.19999,0.200067]$. }    
  \includegraphics[scale=1]{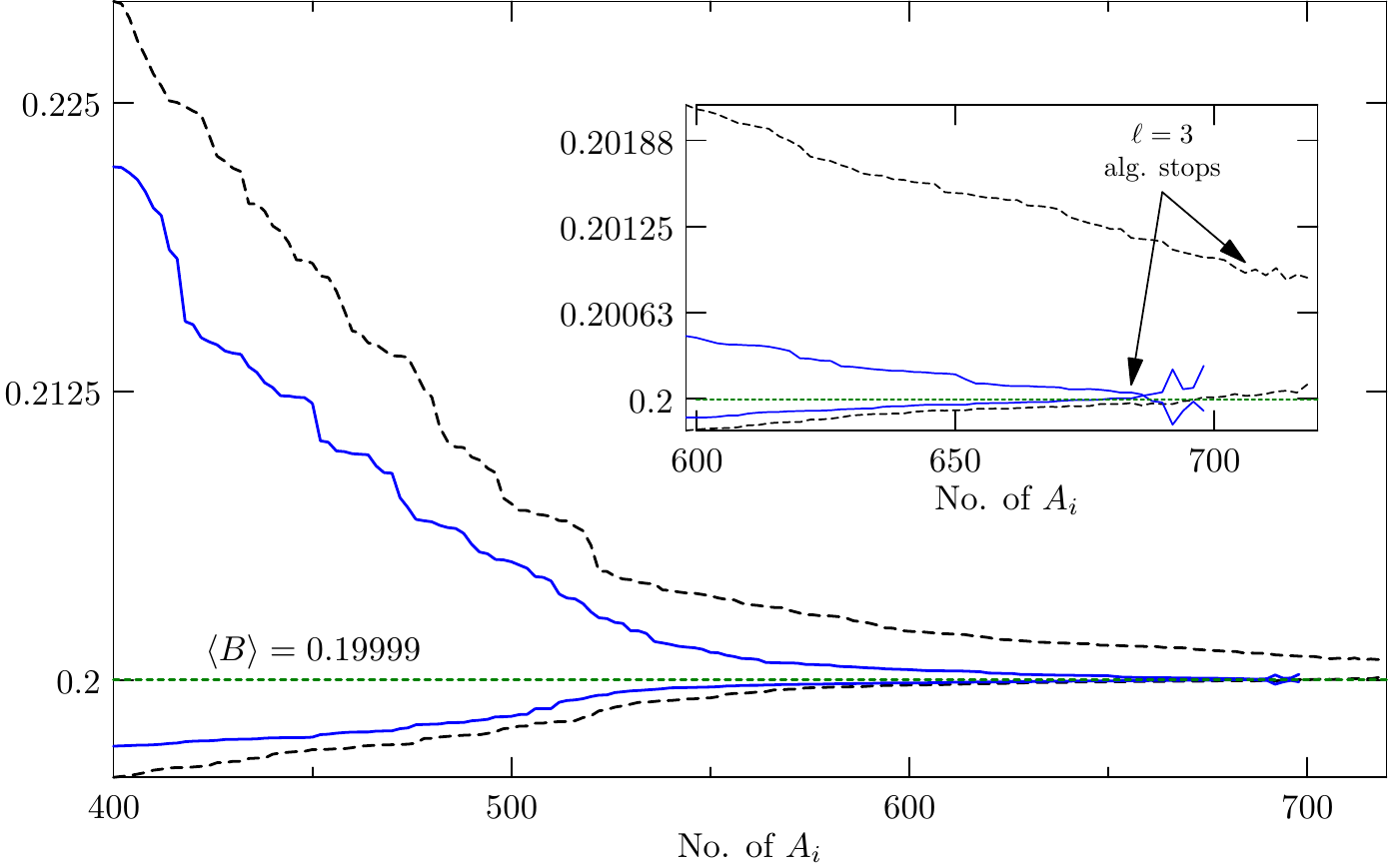}
\end{figure*}

\begin{widetext}
  
\begin{table}[h]
\caption{\label{tab:SysA}%
Numerical results for System-A, the 1D AKLT model.}
\begin{ruledtabular}
\begin{tabular}{cccc}
 Obs &
 \multicolumn{1}{c}{\textrm{exact $\av{B}$}}&
 \multicolumn{1}{c}{\textrm{basic $\ell=3$}}&
 \multicolumn{1}{c}{\textrm{basic $\ell=4$}} \\
\hline \\
  $\textrm{Random}_1$ &  $0.19532$ &
  $[0.19027, 0.20074]\ (0.005) $ &
  $[0.19360, 0.19708]\ (0.002) $ \\
  $\textrm{Random}_2$ &  $0.23058$ &
  $[0.22024, 0.24134]\ (0.01) $ &
  $[0.22709, 0.23412]\ (0.003) $ \\
  $\textrm{Random}_3$ &  $0.27338$ &
  $[0.26209, 0.28511]\ (0.01) $ &
  $[0.26957, 0.27724]\ (0.004) $ 
\end{tabular}
\end{ruledtabular}
\end{table}

\begin{table}[h]
\caption{\label{tab:SysB}%
Numerical results for System-B, the transverse Ising model.}
\begin{ruledtabular}
\begin{tabular}{cccccc}
 Obs &
 \multicolumn{1}{c}{\textrm{exact $\av{B}$}}&
 \multicolumn{1}{c}{\textrm{basic $\ell=3$}}&
 \multicolumn{1}{c}{\textrm{CGO $\ell=3$}}&
 \multicolumn{1}{c}{\textrm{basic $\ell=4$}}&
 \multicolumn{1}{c}{\textrm{CGO $\ell=4$}}\\
\hline \\
  $P_x P_x$ &  $0.38258$ &
  $[0.00197, 0.92964]\ (0.46) $ &
  $[0.37965, 0.38564]\ (0.003) $ &
  $[0.00697, 0.89285]\ (0.44) $ &
  $[0.38189, 0.38284]\ (0.0005) $ \\
  $P_z P_z$ &  $0.79166$ &
  $[0.24231, 0.94085]\ (0.35) $ &
  $[0.79089, 0.79101]\ (6\times 10^{-5})$ &
  $[0.40266, 0.91528]\ (0.25) $ &
  $[0.79166, 0.79234]\ (0.0003) $ \\
  Random &  $0.14391$ &
  $[0.00484, 0.46019]\ (0.23) $ &
  $[0.14376, 0.14377]\ (4\times 10^{-6})$ &
  $[0.01226, 0.33495]\ (0.16) $ &
  $[0.14384, 0.14414]\ (0.0001) $
\end{tabular}
\end{ruledtabular}
\end{table}

\begin{table}[h]
\caption{\label{tab:SysC}%
Numerical results for System-C, the transverse XY model.}
\begin{ruledtabular}
\begin{tabular}{cccccc}
 Obs &
 \multicolumn{1}{c}{\textrm{exact $\av{B}$}}&
 \multicolumn{1}{c}{\textrm{basic $\ell=3$}}&
 \multicolumn{1}{c}{\textrm{CGO $\ell=3$}}&
 \multicolumn{1}{c}{\textrm{basic $\ell=4$}}&
 \multicolumn{1}{c}{\textrm{CGO $\ell=4$}}\\
\hline \\
  $P_x P_x$ &  $0.20054$ &
  $[0.01286, 0.71057]\ (0.35) $ &
  $[0.19856, 0.20321]\ (0.002) $ &
  $[0.03376, 0.58240]\ (0.27) $ &
  $[0.20037, 0.20059]\ (0.0001) $ \\
  $P_z P_z$ &  $ 0.90343$ &
  $[0.21372, 0.98552]\ (0.39) $ &
  $[0.89114, 0.91025]\ (0.009)$ &
  $[0.38477, 0.97288]\ (0.29) $ &
  $[0.90312, 0.90416]\ (0.0005) $ \\
  Random &  $0.27120$ &
  $[0.02852, 0.57259]\ (0.27) $ &
  $[0.26849, 0.27271]\ (0.002)$ &
  $[0.06861, 0.48125]\ (0.20) $ &
  $[0.27111, 0.27131]\ (0.0001) $
\end{tabular}
\end{ruledtabular}
\end{table}

\begin{table}[h]
\caption{\label{tab:SysD}%
Numerical results for System-D, the XY model with random field.}
\begin{ruledtabular}
\begin{tabular}{cccccc}
 Obs &
 \multicolumn{1}{c}{\textrm{exact $\av{B}$}}&
 \multicolumn{1}{c}{\textrm{basic $\ell=3$}}&
 \multicolumn{1}{c}{\textrm{CGO $\ell=3$}}&
 \multicolumn{1}{c}{\textrm{basic $\ell=4$}}&
 \multicolumn{1}{c}{\textrm{CGO $\ell=4$}}\\
\hline \\
  $P_x P_x$ &  $0.19999$ &
  $[0.01222, 0.71539]\ (0.36) $ &
  $[0.19997, 0.20091]\ (0.0005) $ &
  $[0.02336, 0.66029]\ (0.32) $ &
  $[0.19999, 0.20007]\ ( 4\times 10^{-5}) $ \\
  $P_z P_z$ &  $ 0.89168$ &
  $[0.20769, 0.98444]\ (0.39) $ &
  $[0.88757, 0.89180]\ (0.002)$ &
  $[0.19351, 0.97041]\ (0.39) $ &
  $[0.89144, 0.89168]\ (0.0001) $ \\
  Random &  $0.27118$ &
  $[0.02703, 0.42696]\ (0.2) $ &
  $[0.27025, 0.27116]\ (0.0005)$ &
  $[0.03769, 0.49442]\ (0.23) $ &
  $[0.27109, 0.27116]\ (4\times 10^{-5})$
\end{tabular}
\end{ruledtabular}
\end{table}
\end{widetext}

While system A is frustration-free, systems B,C,D are frustrated
 with substantially smaller spectral gaps and larger correlation
 lengths. Their ground states were found using a standard DMRG
 program with a constant bond dimension $D$. We first verified that
 the local expectation values do not change by more than $10^{-5}$
 when passing from $D=6$ to $D=20$. This indicated that $D=6$ is a
 good enough bond dimension for our systems. However, in order to
 suppress as much as possible the finite truncation errors, we first
 computed the ground state as an MPS with $D=20$ and then defined
 the subspace $V_L$ by truncating the bonds at the two cuts from
 index value $7$ onward. All together, this defined a subspace $V_L$
 of dimension $q=6\times 6=36$. The essential difference between
 this procedure and directly using the $D=6$ MPS is that now the 36
 vectors $\{\ket{I_\alpha}\}$ that span $V_L$ are resolved
 \emph{inside} the ball $L$ using $D=20$ instead of $D=6$.

In all systems we used 2-spin observables that acted on spin numbers
50,51 in the middle of the chain. For system-A we used three random
projectors. Each projector was chosen by randomly picking a
2-dimensional subspace out of the 9-dimensional subspace of two
spin-1 particles. In the B,C,D systems we used the three projectors:
(i) $P_xP_x$ projects to the product state of two spin up in the
$\hat{x}$ direction, i.e., $P_xP_x = \ket{+}\bra{+}\otimes
\ket{+}\bra{+}$, (ii) $P_zP_z=\ket{0}\bra{0}\otimes\ket{0}\bra{0}$
is identical to $P_x P_x$ but in the $\hat{z}$ direction, and (iii)
a projector to a random (pure) state in the 4-dimensional
space of the two spin $1/2$ particles.

For each system and each observable we applied the basic and CGO
algorithms with $\ell=3$, that corresponds to a ball $L$ of $8$
spins, and $\ell=4$, with $10$ spins. Since System-A is
frustration-free we only applied the basic algorithm to it. 

In the frustrated systems our SDP calculations were performed using
the open-source package \texttt{SDPA}\cc{ref:SDPAbook, ref:SDPArep}.
As noted in \Sec{sec:constant-D}, when working with constant $D$,
\Eq{eq:rhoL-constraint} is usually over-determined, and therefore
one should not use a complete spanning basis of $A_i$ matrices.
Following the idea from that section, we used a set of random
anti-Hermitian matrices $A_i$ which were created as random MPO of
bond dimension $D=20$. Then we ran the SDP for the lower and upper
bounds on $\av{B}$, gradually increasing the number of $A_i$. 

We used the following heuristics to decide when to stop adding more
$A_i$'s and
read the final result. First, at every step, we calculated the
difference between the upperbound and the lowerbound to $\av{B}$,
and this served as a natural error scale. If the difference became
negative (i.e., lowerbound was greater than the upperbound), we took
the results of the previous step. Additionally, we stopped if the
graph started to ``wiggle'': if at one step the upperbound or
lowerbound were worsened with respect to the previous step by more
than 1\% of the error length scale, we stopped and used the result
of the previous step. In all runs (in systems B,C,D) this procedure
resulted in an SDP estimate that is based on about 520-710 $A_i$
matrices. 

A typical evolution of the upper/lower bounds for $\ell=3,4$ as the
number of $A_i$ increases is presented in
\Fig{fig:sysD-PxPx-converge}. The graph shows the evolution of the
bounds in the case of system D (the XY model with random transverse
field) with the random projector.

The full numerical results for the four systems are summarized in
Tables~\ref{tab:SysA},\ref{tab:SysB},\ref{tab:SysC},\ref{tab:SysD}.
Table~\ref{tab:SysA} presents the $\ell=3,4$ basic algorithm results
for system A, the frustration-free AKLT model. The other three
tables are for the frustrated models B,C,D, and they also include
the CGO results. For each test we present the lowerbound and the
upperbound in the form of a $[k_{min},k_{max}]$ range, as well as
the estimated error $\Delta\EqDef \frac{1}{2}(k_{max}-k_{min})$.

In the frustrated case, the $\ell=3,4$ radii of the local
environment are well below the range where decay of correlation
should be felt. It is not surprising then that the basic algorithm 
in these cases performs very badly, and is unable to
give a single meaningful digit in the approximation of $\av{B}$. At
the same time, quite remarkably, the CGO algorithm manages to recover
3--4 digits (and often 4--5 digits) from $\av{B}$. Moreover, these
results are always better than their equivalent basic results for
the frustration-free AKLT model, despite having a much larger
correlation length ($\xi=8.8$ and $\xi=4.3$ vs. $\xi=0.9$ for the
AKLT) and a much larger $V_L$ subspace ($q=36$ vs. $q=4$ for the
AKLT).

\section{Conclusions}
\label{sec:discussion}

We have presented a new approach to approximate the expectation
value of a local operator given a PEPS tensor-network. Our initial
observation is that while this is generally a computationally hard
task (\#P-hard), it is not necessarily the case for PEPS that
represent gapped ground states, since these states have much more
structure. Our approach circumvents the exponentially expensive
contraction of the environment by estimating the local expectation
value using only a local patch of the PEPS tensor-network around it.
We have presented two algorithms to accomplish that. The basic
algorithm simply diagonalizes the observable in the subspace $V_L$
and uses its extreme eigenvalues as upper/lower bounds for the
expectation value. We have argued that this range is expected to
converge as the radius of the local patch increases; it should in
particular be useful for frustration-free systems. The second
algorithm, CGO, builds upon the basic algorithm, but uses optimization
over commutators to narrow down the eigenvalue range significantly.
For it to work, the system has to be frustrated.

As demonstrated by our numerical tests, both algorithms allow one to
extract a significant amount of information about the local
expectation value using only a local patch of the surrounding tensor
network. Contrary to what one might have expected, it seems that in
frustrated systems much more information can be extracted from a
local patch. In hindsight, this is reasonable: in frustrated systems
the ground state is not a common eigenstate of all local
terms~\footnote{Strictly speaking, one can have a frustrated system
in which the groundstate is a common eigenstate -- for example, a
frustrated classical system. However, in some sense these are not
`genuine' frustrated systems, since they can be trivially
transformed into a frustration-free system by by locally subtracting
an identity term from every local term in the Hamiltonian. This
would make the system frustration-free and only shift its spectrum
by a constant.}; only once we consider all local terms, we satisfy
the eigenstate equation. This implies that the action of one local
term must somehow cancel the action of other local terms, and this
inter-dependence leads to many non-trivial constraints that can be
exploited to recover $\rho_L$ from its underlying subspace (see
\Cor{cor:rhoL-constraint}).

Our work leaves many open questions for future research. From a
numerical point of view, the most immediate question we would like
to answer is how effective the CGO algorithm is in 2D. As it stands
now, repeating the same procedure that we have used in the 1D case
is impractical. Tentatively, for the CGO algorithm to work, one
needs a high enough $D$ that would yield good approximation to the
groundstate of a frustrated system, and in addition the dimension of
the $V_L$ subspace should be smaller than the physical dimension of
the spins inside $L$ -- i.e., the area law should be non-trivially
felt. Consider, for example, the rectangular grid in
\Fig{fig:Lregion}. There, we have $\dim V_L = D^{2(4\ell+3)}$ while
the physical dimension goes like $d^{2(\ell+1)^2}$. So if we
consider a $D=4$ PEPS over $d=2$ spins, the physical dimension wins
over $\dim V_L$ at $\ell=7$, and reasonable results are expected at
$\ell=8$ or beyond. This would mean one has to work with matrices of
size $2^{162}\times 2^{162}$, an impossible task.  It is therefore
clear that the CGO algorithm cannot be used directly in 2D.
Nevertheless, it might still be possible to apply it in approximate
manner. For example, by projecting the system to a random subspace
and using concentration results like the Johnson-Lindenstrauss
lemma\cc{ref:JL1984}, or perhaps, by exploiting the symmetries in
the problem (like translation invariance) in a clever way.

It would be also interesting to see if some aspects of CGO can be
used in existing PEPS algorithms. The constraints of
\Cor{cor:rhoL-constraint} may be useful in improving the
simple-update method; even if we deal with small regions $L$ that do
not provide a unique answer for $\rho_L$, satisfying the constraints
in \Cor{cor:rhoL-constraint} may improve upon the simple
product-state approximation of the environment. In addition, they
might be useful for increasing the numerical stability of current
algorithms. Finally, from a numerical point of view it would be
interesting to see if \Cor{cor:rhoL-constraint} could be used to
verify that we have an eigenstate at our hands. This might be
useful, for example, for settling the nature of the groundstate of
the antiferromagnetic Heisenberg model on the kagome
lattice\cc{ref:EV2010-Kagome, ref:YHW2011-Kagome,
ref:DMS2012-Kagome, ref:IBSP2013-Kagome, ref:CCAW2013-Kagome}.

From a more theoretical point of view, it would be interesting to
gain better understanding of the CGO algorithm, and in particular
the implications of the constraints in \Cor{cor:rhoL-constraint}.
For example, can we find sufficient conditions to when the algorithm
gives good approximations (i.e., narrow $[k_{min},k_{max}]$ range)?
From a computational complexity point of view, this might be a step
in showing that PEPS of ground states can serve as an NP witness. A
much more ambitious goal in that direction would be to show that the
complexity of the local Hamiltonian problem for gapped Hamiltonians
on a regular lattice with a unique ground state is inside NP.


\begin{acknowledgments}
  We thank T.~Kuwahara, I.~Latorre, I.~Cirac, D.~ P\'erez-Garc\'ia,
  and in particular D.~Aharonov for many helpful discussions. We
  thank T.~Nishino for bringing to our attention early references for
  the uses of PEPS tensor-network and the CTM.
  A.A was supported by Core grants of Centre for Quantum
  Technologies, Singapore. Research at the Centre for Quantum
  Technologies is funded by the Singapore Ministry of Education and
  the National Research Foundation, also through the Tier 3 Grant
  random numbers from quantum processes.
\end{acknowledgments}

\appendix

\section{Proof of \Lem{lem:FF-AGSP}}

The proof of \Lem{lem:FF-AGSP} utilizes the
detectability-lemma\cc{ref:AAVZ2011-DL} (see also
\Ref{ref:AAV2016-DLD}), and to a large extent already appeared in
\Ref{ref:AAVZ2011-DL}. We give the full proof here for sake of
completeness. 

We start by introducing the detectability lemma. Given a local
Hamiltonian $H=\sum_i h_i$, we 
define for every $h_i$ a projector $P_i$ that projects
onto its local ground space.  Next we partition the $P_i$ terms into
$g$ `layers', such that each layer is a collection of projectors
with disjoint support. For example, a 1D chain can be partitioned
into two layers -- one for projectors on $(i,i+1)$ with
even $i$'s, and the second for odd $i$'s. Then for every layer
$\ell$, we define the projector $\Pi_\ell$, which is the product of
all the $P_i$ projectors in that layer; it is the projector onto the
common ground space of all local terms in that layer. The
detectability lemma operator, $\DL(H)$ is then defined to be the
product of the layer projectors: $\DL(H)\EqDef\Pi_1\cdots \Pi_g$.
Because the system is frustration free, it follows that 
$\DL(H)\ket{\gs}=\ket{\gs}$. The detectability lemma states that
if both the number of layers and the spectral gap of $H$ 
are of $\orderof{1}$, then for every state $\ket{\psi}$ orthogonal
to the ground space, 
$\norm{\DL(H)\ket{\gs^\perp}}^2 \le c<1$ for some constant $c$. 
Therefore, taking $\ell$ powers of the detectability lemma operator,
we get an exponentially good approximation to the ground space
projector:
\begin{align}
\label{eq:DL}
  \norm{\DL^\ell(H) - \ket{\gs}\bra{\gs}} \le
  e^{-\orderof{\ell}} \,.
\end{align}

Let us now return to the proof of the lemma. It consists of two
steps. The first step is to notice that 
\begin{align}
\label{eq:light-cone}
  B_L\ket{\gs} = P_V B\ket{\gs} = P_V \DL^{\ell'}(H) B\ket{\gs} \,,
\end{align}
where $\ell'=\orderof{\ell}$. The first equality follows from
definition. The second equality follows from writing
$\DL(H)^{\ell'}$ in terms of the local $P_i$ projectors, and
noticing that each such $P_i$ projector can be either `pulled' from
the $P_V$ projector or from the ground state $\ket{\gs}$. This is
illustrated in \Fig{fig:light-cone} for the 1D case, see
\Ref{ref:AAVZ2011-DL} for more details.
\begin{figure}
  \includegraphics[scale=0.9]{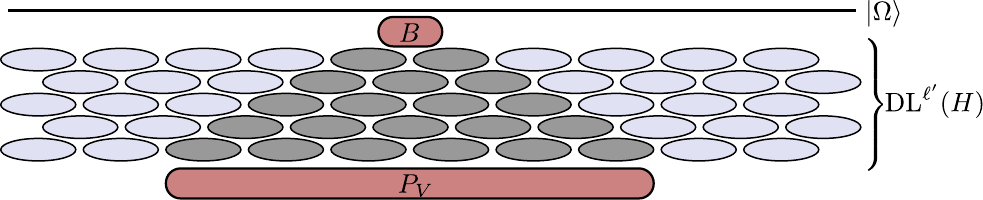}
  \caption{\label{fig:light-cone} Illustration of the second
  inequality in \Eq{eq:light-cone} in the 1D case. The ovals
  represent the $P_i$ projectors in $\DL^{\ell'}(H)$, organized in
  layers. The dark gray ovals are projectors that can be `pulled'
  out from the $P_V$ operator (here, on the bottom). They form a sort of
  casual `light-cone' defined by the local operator $B$. The rest of
  the projectors (represented by light-gray ovals) can be pulled out
  of the ground state $\ket{\gs}$ (here, at the top). }
\end{figure}

The second step uses \eqref{eq:DL} to replace $\DL^{\ell'}(H)$ by
$\ket{\gs}\bra{\gs} + e^{-\orderof{\ell'}} =\ket{\gs}\bra{\gs} +
e^{-\orderof{\ell}}$, giving us
\begin{align}
  B_L\ket{\gs} &= P_V\cdot \ket{\gs}\bra{\gs}\cdot B \ket{\gs} 
      + e^{-\orderof{\ell}} \\
    &= \av{B}\ket{\gs} + e^{-\orderof{\ell}}\,. \nonumber
\end{align}
This proves the lemma.

\bibliography{LOCPEPS}

\begin{thebibliography}{49}%
\makeatletter
\providecommand \@ifxundefined [1]{%
 \@ifx{#1\undefined}
}%
\providecommand \@ifnum [1]{%
 \ifnum #1\expandafter \@firstoftwo
 \else \expandafter \@secondoftwo
 \fi
}%
\providecommand \@ifx [1]{%
 \ifx #1\expandafter \@firstoftwo
 \else \expandafter \@secondoftwo
 \fi
}%
\providecommand \natexlab [1]{#1}%
\providecommand \enquote  [1]{``#1''}%
\providecommand \bibnamefont  [1]{#1}%
\providecommand \bibfnamefont [1]{#1}%
\providecommand \citenamefont [1]{#1}%
\providecommand \href@noop [0]{\@secondoftwo}%
\providecommand \href [0]{\begingroup \@sanitize@url \@href}%
\providecommand \@href[1]{\@@startlink{#1}\@@href}%
\providecommand \@@href[1]{\endgroup#1\@@endlink}%
\providecommand \@sanitize@url [0]{\catcode `\\12\catcode `\$12\catcode
  `\&12\catcode `\#12\catcode `\^12\catcode `\_12\catcode `\%12\relax}%
\providecommand \@@startlink[1]{}%
\providecommand \@@endlink[0]{}%
\providecommand \url  [0]{\begingroup\@sanitize@url \@url }%
\providecommand \@url [1]{\endgroup\@href {#1}{\urlprefix }}%
\providecommand \urlprefix  [0]{URL }%
\providecommand \Eprint [0]{\href }%
\providecommand \doibase [0]{http://dx.doi.org/}%
\providecommand \selectlanguage [0]{\@gobble}%
\providecommand \bibinfo  [0]{\@secondoftwo}%
\providecommand \bibfield  [0]{\@secondoftwo}%
\providecommand \translation [1]{[#1]}%
\providecommand \BibitemOpen [0]{}%
\providecommand \bibitemStop [0]{}%
\providecommand \bibitemNoStop [0]{.\EOS\space}%
\providecommand \EOS [0]{\spacefactor3000\relax}%
\providecommand \BibitemShut  [1]{\csname bibitem#1\endcsname}%
\let\auto@bib@innerbib\@empty
\bibitem [{\citenamefont {Or\'us}(2014)}]{ref:Orus2014-TN-Intro}%
  \BibitemOpen
  \bibfield  {author} {\bibinfo {author} {\bibfnamefont {R.}~\bibnamefont
  {Or\'us}},\ }\href {\doibase 10.1016/j.aop.2014.06.013} {\bibfield  {journal}
  {\bibinfo  {journal} {Annals of Physics}\ }\textbf {\bibinfo {volume}
  {349}},\ \bibinfo {pages} {117 } (\bibinfo {year} {2014})}\BibitemShut
  {NoStop}%
\bibitem [{\citenamefont {White}(1992)}]{ref:White1992-DMRG1}%
  \BibitemOpen
  \bibfield  {author} {\bibinfo {author} {\bibfnamefont {S.~R.}\ \bibnamefont
  {White}},\ }\href {\doibase 10.1103/PhysRevLett.69.2863} {\bibfield
  {journal} {\bibinfo  {journal} {Phys. Rev. Lett.}\ }\textbf {\bibinfo
  {volume} {69}},\ \bibinfo {pages} {2863} (\bibinfo {year}
  {1992})}\BibitemShut {NoStop}%
\bibitem [{\citenamefont {White}(1993)}]{ref:White1993-DMRG2}%
  \BibitemOpen
  \bibfield  {author} {\bibinfo {author} {\bibfnamefont {S.~R.}\ \bibnamefont
  {White}},\ }\href {\doibase 10.1103/PhysRevB.48.10345} {\bibfield  {journal}
  {\bibinfo  {journal} {Phys. Rev. B}\ }\textbf {\bibinfo {volume} {48}},\
  \bibinfo {pages} {10345} (\bibinfo {year} {1993})}\BibitemShut {NoStop}%
\bibitem [{\citenamefont {{\"O}stlund}\ and\ \citenamefont
  {Rommer}(1995)}]{ref:OR1995-DMRG-as-MPS1}%
  \BibitemOpen
  \bibfield  {author} {\bibinfo {author} {\bibfnamefont {S.}~\bibnamefont
  {{\"O}stlund}}\ and\ \bibinfo {author} {\bibfnamefont {S.}~\bibnamefont
  {Rommer}},\ }\href@noop {} {\bibfield  {journal} {\bibinfo  {journal}
  {Physical review letters}\ }\textbf {\bibinfo {volume} {75}},\ \bibinfo
  {pages} {3537} (\bibinfo {year} {1995})}\BibitemShut {NoStop}%
\bibitem [{\citenamefont {Rommer}\ and\ \citenamefont
  {{\"O}stlund}(1997)}]{ref:RO1997-DMRG-as-MPS2}%
  \BibitemOpen
  \bibfield  {author} {\bibinfo {author} {\bibfnamefont {S.}~\bibnamefont
  {Rommer}}\ and\ \bibinfo {author} {\bibfnamefont {S.}~\bibnamefont
  {{\"O}stlund}},\ }\href@noop {} {\bibfield  {journal} {\bibinfo  {journal}
  {Physical Review B}\ }\textbf {\bibinfo {volume} {55}},\ \bibinfo {pages}
  {2164} (\bibinfo {year} {1997})}\BibitemShut {NoStop}%
\bibitem [{\citenamefont {{Verstraete}}\ and\ \citenamefont
  {{Cirac}}(2004)}]{ref:VC2004-PEPS}%
  \BibitemOpen
  \bibfield  {author} {\bibinfo {author} {\bibfnamefont {F.}~\bibnamefont
  {{Verstraete}}}\ and\ \bibinfo {author} {\bibfnamefont {J.~I.}\ \bibnamefont
  {{Cirac}}},\ }\href {http://arxiv.org/abs/cond-mat/0407066} {\bibfield
  {journal} {\bibinfo  {journal} {eprint arXiv:cond-mat/0407066}\ } (\bibinfo
  {year} {2004})},\ \Eprint {http://arxiv.org/abs/cond-mat/0407066}
  {cond-mat/0407066} \BibitemShut {NoStop}%
\bibitem [{\citenamefont {Verstraete}\ \emph {et~al.}(2008)\citenamefont
  {Verstraete}, \citenamefont {Murg},\ and\ \citenamefont
  {Cirac}}]{ref:VMC2008-PEPS-rev}%
  \BibitemOpen
  \bibfield  {author} {\bibinfo {author} {\bibfnamefont {F.}~\bibnamefont
  {Verstraete}}, \bibinfo {author} {\bibfnamefont {V.}~\bibnamefont {Murg}}, \
  and\ \bibinfo {author} {\bibfnamefont {J.~I.}\ \bibnamefont {Cirac}},\ }\href
  {\doibase 10.1080/14789940801912366} {\bibfield  {journal} {\bibinfo
  {journal} {Advances in Physics}\ }\textbf {\bibinfo {volume} {57}},\ \bibinfo
  {pages} {143} (\bibinfo {year} {2008})},\ \Eprint
  {http://arxiv.org/abs/arXiv:0907.2796} {arXiv:0907.2796} \BibitemShut
  {NoStop}%
\bibitem [{\citenamefont {Sierra}\ and\ \citenamefont
  {Martin-Delgado}(1998)}]{ref:SM1998-DMRG-CFT}%
  \BibitemOpen
  \bibfield  {author} {\bibinfo {author} {\bibfnamefont {G.}~\bibnamefont
  {Sierra}}\ and\ \bibinfo {author} {\bibfnamefont {M.~A.}\ \bibnamefont
  {Martin-Delgado}},\ }\href@noop {} {\bibfield  {journal} {\bibinfo  {journal}
  {arXiv preprint cond-mat/9811170}\ } (\bibinfo {year} {1998})},\ \Eprint
  {http://arxiv.org/abs/arXiv:cond-mat/9811170} {arXiv:cond-mat/9811170}
  \BibitemShut {NoStop}%
\bibitem [{\citenamefont {Hieida}\ \emph {et~al.}(1999)\citenamefont {Hieida},
  \citenamefont {Okunishi},\ and\ \citenamefont {Akutsu}}]{ref:HOK1999:NRG2D}%
  \BibitemOpen
  \bibfield  {author} {\bibinfo {author} {\bibfnamefont {Y.}~\bibnamefont
  {Hieida}}, \bibinfo {author} {\bibfnamefont {K.}~\bibnamefont {Okunishi}}, \
  and\ \bibinfo {author} {\bibfnamefont {Y.}~\bibnamefont {Akutsu}},\ }\href
  {http://stacks.iop.org/1367-2630/1/i=1/a=007} {\bibfield  {journal} {\bibinfo
   {journal} {New Journal of Physics}\ }\textbf {\bibinfo {volume} {1}},\
  \bibinfo {pages} {7} (\bibinfo {year} {1999})},\ \Eprint
  {http://arxiv.org/abs/arXiv:cond-mat/9901155} {arXiv:cond-mat/9901155}
  \BibitemShut {NoStop}%
\bibitem [{\citenamefont {Nishino}\ \emph {et~al.}(2001)\citenamefont
  {Nishino}, \citenamefont {Hieida}, \citenamefont {Okunishi}, \citenamefont
  {Maeshima}, \citenamefont {Akutsu},\ and\ \citenamefont
  {Gendiar}}]{ref:NHYOMAG2001-TPS}%
  \BibitemOpen
  \bibfield  {author} {\bibinfo {author} {\bibfnamefont {T.}~\bibnamefont
  {Nishino}}, \bibinfo {author} {\bibfnamefont {Y.}~\bibnamefont {Hieida}},
  \bibinfo {author} {\bibfnamefont {K.}~\bibnamefont {Okunishi}}, \bibinfo
  {author} {\bibfnamefont {N.}~\bibnamefont {Maeshima}}, \bibinfo {author}
  {\bibfnamefont {Y.}~\bibnamefont {Akutsu}}, \ and\ \bibinfo {author}
  {\bibfnamefont {A.}~\bibnamefont {Gendiar}},\ }\href {\doibase
  10.1143/PTP.105.409} {\bibfield  {journal} {\bibinfo  {journal} {Progress of
  Theoretical Physics}\ }\textbf {\bibinfo {volume} {105}},\ \bibinfo {pages}
  {409} (\bibinfo {year} {2001})},\ \Eprint
  {http://arxiv.org/abs/arXiv:cond-mat/0011103} {arXiv:cond-mat/0011103}
  \BibitemShut {NoStop}%
\bibitem [{\citenamefont {Shi}\ \emph {et~al.}(2006)\citenamefont {Shi},
  \citenamefont {Duan},\ and\ \citenamefont {Vidal}}]{ref:SDV2006-TreeTN}%
  \BibitemOpen
  \bibfield  {author} {\bibinfo {author} {\bibfnamefont {Y.-Y.}\ \bibnamefont
  {Shi}}, \bibinfo {author} {\bibfnamefont {L.-M.}\ \bibnamefont {Duan}}, \
  and\ \bibinfo {author} {\bibfnamefont {G.}~\bibnamefont {Vidal}},\ }\href
  {\doibase 10.1103/PhysRevA.74.022320} {\bibfield  {journal} {\bibinfo
  {journal} {Phys. Rev. A}\ }\textbf {\bibinfo {volume} {74}},\ \bibinfo
  {pages} {022320} (\bibinfo {year} {2006})}\BibitemShut {NoStop}%
\bibitem [{\citenamefont {Vidal}(2008)}]{ref:Vidal2008-MERA}%
  \BibitemOpen
  \bibfield  {author} {\bibinfo {author} {\bibfnamefont {G.}~\bibnamefont
  {Vidal}},\ }\href {\doibase 10.1103/PhysRevLett.101.110501} {\bibfield
  {journal} {\bibinfo  {journal} {Phys. Rev. Lett.}\ }\textbf {\bibinfo
  {volume} {101}},\ \bibinfo {pages} {110501} (\bibinfo {year}
  {2008})}\BibitemShut {NoStop}%
\bibitem [{\citenamefont {Schuch}\ \emph {et~al.}(2008)\citenamefont {Schuch},
  \citenamefont {Wolf}, \citenamefont {Verstraete},\ and\ \citenamefont
  {Cirac}}]{ref:SWFC2008-String-States}%
  \BibitemOpen
  \bibfield  {author} {\bibinfo {author} {\bibfnamefont {N.}~\bibnamefont
  {Schuch}}, \bibinfo {author} {\bibfnamefont {M.~M.}\ \bibnamefont {Wolf}},
  \bibinfo {author} {\bibfnamefont {F.}~\bibnamefont {Verstraete}}, \ and\
  \bibinfo {author} {\bibfnamefont {J.~I.}\ \bibnamefont {Cirac}},\ }\href
  {\doibase 10.1103/PhysRevLett.100.040501} {\bibfield  {journal} {\bibinfo
  {journal} {Phys. Rev. Lett.}\ }\textbf {\bibinfo {volume} {100}},\ \bibinfo
  {pages} {040501} (\bibinfo {year} {2008})}\BibitemShut {NoStop}%
\bibitem [{\citenamefont {Xie}\ \emph {et~al.}(2014)\citenamefont {Xie},
  \citenamefont {Chen}, \citenamefont {Yu}, \citenamefont {Kong}, \citenamefont
  {Normand},\ and\ \citenamefont {Xiang}}]{ref:XCYFNX2014-PESS}%
  \BibitemOpen
  \bibfield  {author} {\bibinfo {author} {\bibfnamefont {Z.~Y.}\ \bibnamefont
  {Xie}}, \bibinfo {author} {\bibfnamefont {J.}~\bibnamefont {Chen}}, \bibinfo
  {author} {\bibfnamefont {J.~F.}\ \bibnamefont {Yu}}, \bibinfo {author}
  {\bibfnamefont {X.}~\bibnamefont {Kong}}, \bibinfo {author} {\bibfnamefont
  {B.}~\bibnamefont {Normand}}, \ and\ \bibinfo {author} {\bibfnamefont
  {T.}~\bibnamefont {Xiang}},\ }\href {\doibase 10.1103/PhysRevX.4.011025}
  {\bibfield  {journal} {\bibinfo  {journal} {Phys. Rev. X}\ }\textbf {\bibinfo
  {volume} {4}},\ \bibinfo {pages} {011025} (\bibinfo {year}
  {2014})}\BibitemShut {NoStop}%
\bibitem [{\citenamefont {Schuch}\ \emph {et~al.}(2007)\citenamefont {Schuch},
  \citenamefont {Wolf}, \citenamefont {Verstraete},\ and\ \citenamefont
  {Cirac}}]{ref:SWVC2007-PEPS-Complexity}%
  \BibitemOpen
  \bibfield  {author} {\bibinfo {author} {\bibfnamefont {N.}~\bibnamefont
  {Schuch}}, \bibinfo {author} {\bibfnamefont {M.~M.}\ \bibnamefont {Wolf}},
  \bibinfo {author} {\bibfnamefont {F.}~\bibnamefont {Verstraete}}, \ and\
  \bibinfo {author} {\bibfnamefont {J.~I.}\ \bibnamefont {Cirac}},\ }\href
  {\doibase 10.1103/PhysRevLett.98.140506} {\bibfield  {journal} {\bibinfo
  {journal} {Phys. Rev. Lett.}\ }\textbf {\bibinfo {volume} {98}},\ \bibinfo
  {pages} {140506} (\bibinfo {year} {2007})}\BibitemShut {NoStop}%
\bibitem [{\citenamefont {Nishino}\ and\ \citenamefont
  {Okunishi}(1996)}]{ref:NO1996-CTM0}%
  \BibitemOpen
  \bibfield  {author} {\bibinfo {author} {\bibfnamefont {T.}~\bibnamefont
  {Nishino}}\ and\ \bibinfo {author} {\bibfnamefont {K.}~\bibnamefont
  {Okunishi}},\ }\href@noop {} {\bibfield  {journal} {\bibinfo  {journal}
  {Journal of the Physical Society of Japan}\ }\textbf {\bibinfo {volume}
  {65}},\ \bibinfo {pages} {891} (\bibinfo {year} {1996})},\ \Eprint
  {http://arxiv.org/abs/arXiv:cond-mat/9507087} {arXiv:cond-mat/9507087}
  \BibitemShut {NoStop}%
\bibitem [{\citenamefont {Or\'us}\ and\ \citenamefont
  {Vidal}(2009)}]{ref:OV2009-CTM1}%
  \BibitemOpen
  \bibfield  {author} {\bibinfo {author} {\bibfnamefont {R.}~\bibnamefont
  {Or\'us}}\ and\ \bibinfo {author} {\bibfnamefont {G.}~\bibnamefont {Vidal}},\
  }\href {\doibase 10.1103/PhysRevB.80.094403} {\bibfield  {journal} {\bibinfo
  {journal} {Phys. Rev. B}\ }\textbf {\bibinfo {volume} {80}},\ \bibinfo
  {pages} {094403} (\bibinfo {year} {2009})},\ \Eprint
  {http://arxiv.org/abs/arXiv:0905.3225} {arXiv:0905.3225} \BibitemShut
  {NoStop}%
\bibitem [{\citenamefont {Or\'us}(2012)}]{ref:Orus2012-CTM2}%
  \BibitemOpen
  \bibfield  {author} {\bibinfo {author} {\bibfnamefont {R.}~\bibnamefont
  {Or\'us}},\ }\href {\doibase 10.1103/PhysRevB.85.205117} {\bibfield
  {journal} {\bibinfo  {journal} {Phys. Rev. B}\ }\textbf {\bibinfo {volume}
  {85}},\ \bibinfo {pages} {205117} (\bibinfo {year} {2012})},\ \Eprint
  {http://arxiv.org/abs/arXiv:1112.4101} {arXiv:1112.4101} \BibitemShut
  {NoStop}%
\bibitem [{\citenamefont {Vanderstraeten}\ \emph {et~al.}(2015)\citenamefont
  {Vanderstraeten}, \citenamefont {Mari\"en}, \citenamefont {Verstraete},\ and\
  \citenamefont {Haegeman}}]{ref:VMVH2015-CTM3}%
  \BibitemOpen
  \bibfield  {author} {\bibinfo {author} {\bibfnamefont {L.}~\bibnamefont
  {Vanderstraeten}}, \bibinfo {author} {\bibfnamefont {M.}~\bibnamefont
  {Mari\"en}}, \bibinfo {author} {\bibfnamefont {F.}~\bibnamefont
  {Verstraete}}, \ and\ \bibinfo {author} {\bibfnamefont {J.}~\bibnamefont
  {Haegeman}},\ }\href {\doibase 10.1103/PhysRevB.92.201111} {\bibfield
  {journal} {\bibinfo  {journal} {Phys. Rev. B}\ }\textbf {\bibinfo {volume}
  {92}},\ \bibinfo {pages} {201111} (\bibinfo {year} {2015})},\ \Eprint
  {http://arxiv.org/abs/arXiv:1507.02151} {arXiv:1507.02151} \BibitemShut
  {NoStop}%
\bibitem [{\citenamefont {Levin}\ and\ \citenamefont
  {Nave}(2007)}]{ref:LN2007-TRG}%
  \BibitemOpen
  \bibfield  {author} {\bibinfo {author} {\bibfnamefont {M.}~\bibnamefont
  {Levin}}\ and\ \bibinfo {author} {\bibfnamefont {C.~P.}\ \bibnamefont
  {Nave}},\ }\href {\doibase 10.1103/PhysRevLett.99.120601} {\bibfield
  {journal} {\bibinfo  {journal} {Phys. Rev. Lett.}\ }\textbf {\bibinfo
  {volume} {99}},\ \bibinfo {pages} {120601} (\bibinfo {year} {2007})},\
  \Eprint {http://arxiv.org/abs/arXiv:cond-mat/0611687}
  {arXiv:cond-mat/0611687} \BibitemShut {NoStop}%
\bibitem [{\citenamefont {Xie}\ \emph {et~al.}(2009)\citenamefont {Xie},
  \citenamefont {Jiang}, \citenamefont {Chen}, \citenamefont {Weng},\ and\
  \citenamefont {Xiang}}]{ref:XJCWX2009-SRG1}%
  \BibitemOpen
  \bibfield  {author} {\bibinfo {author} {\bibfnamefont {Z.~Y.}\ \bibnamefont
  {Xie}}, \bibinfo {author} {\bibfnamefont {H.~C.}\ \bibnamefont {Jiang}},
  \bibinfo {author} {\bibfnamefont {Q.~N.}\ \bibnamefont {Chen}}, \bibinfo
  {author} {\bibfnamefont {Z.~Y.}\ \bibnamefont {Weng}}, \ and\ \bibinfo
  {author} {\bibfnamefont {T.}~\bibnamefont {Xiang}},\ }\href {\doibase
  10.1103/PhysRevLett.103.160601} {\bibfield  {journal} {\bibinfo  {journal}
  {Phys. Rev. Lett.}\ }\textbf {\bibinfo {volume} {103}},\ \bibinfo {pages}
  {160601} (\bibinfo {year} {2009})},\ \Eprint
  {http://arxiv.org/abs/arXiv:0809.0182} {arXiv:0809.0182} \BibitemShut
  {NoStop}%
\bibitem [{\citenamefont {Zhao}\ \emph {et~al.}(2010)\citenamefont {Zhao},
  \citenamefont {Xie}, \citenamefont {Chen}, \citenamefont {Wei}, \citenamefont
  {Cai},\ and\ \citenamefont {Xiang}}]{ref:ZXCWCX2010-SRG2}%
  \BibitemOpen
  \bibfield  {author} {\bibinfo {author} {\bibfnamefont {H.~H.}\ \bibnamefont
  {Zhao}}, \bibinfo {author} {\bibfnamefont {Z.~Y.}\ \bibnamefont {Xie}},
  \bibinfo {author} {\bibfnamefont {Q.~N.}\ \bibnamefont {Chen}}, \bibinfo
  {author} {\bibfnamefont {Z.~C.}\ \bibnamefont {Wei}}, \bibinfo {author}
  {\bibfnamefont {J.~W.}\ \bibnamefont {Cai}}, \ and\ \bibinfo {author}
  {\bibfnamefont {T.}~\bibnamefont {Xiang}},\ }\href {\doibase
  10.1103/PhysRevB.81.174411} {\bibfield  {journal} {\bibinfo  {journal} {Phys.
  Rev. B}\ }\textbf {\bibinfo {volume} {81}},\ \bibinfo {pages} {174411}
  (\bibinfo {year} {2010})},\ \Eprint {http://arxiv.org/abs/arXiv:1002.1405}
  {arXiv:1002.1405} \BibitemShut {NoStop}%
\bibitem [{\citenamefont {Xie}\ \emph {et~al.}(2012)\citenamefont {Xie},
  \citenamefont {Chen}, \citenamefont {Qin}, \citenamefont {Zhu}, \citenamefont
  {Yang},\ and\ \citenamefont {Xiang}}]{ref:XCQZYX2012-HOTRG}%
  \BibitemOpen
  \bibfield  {author} {\bibinfo {author} {\bibfnamefont {Z.~Y.}\ \bibnamefont
  {Xie}}, \bibinfo {author} {\bibfnamefont {J.}~\bibnamefont {Chen}}, \bibinfo
  {author} {\bibfnamefont {M.~P.}\ \bibnamefont {Qin}}, \bibinfo {author}
  {\bibfnamefont {J.~W.}\ \bibnamefont {Zhu}}, \bibinfo {author} {\bibfnamefont
  {L.~P.}\ \bibnamefont {Yang}}, \ and\ \bibinfo {author} {\bibfnamefont
  {T.}~\bibnamefont {Xiang}},\ }\href {\doibase 10.1103/PhysRevB.86.045139}
  {\bibfield  {journal} {\bibinfo  {journal} {Phys. Rev. B}\ }\textbf {\bibinfo
  {volume} {86}},\ \bibinfo {pages} {045139} (\bibinfo {year} {2012})},\
  \Eprint {http://arxiv.org/abs/arXiv:1201.1144} {arXiv:1201.1144} \BibitemShut
  {NoStop}%
\bibitem [{\citenamefont {Pi\ifmmode~\check{z}\else \v{z}\fi{}orn}\ \emph
  {et~al.}(2011)\citenamefont {Pi\ifmmode~\check{z}\else \v{z}\fi{}orn},
  \citenamefont {Wang},\ and\ \citenamefont
  {Verstraete}}]{ref:PWV2011-Single-Layer}%
  \BibitemOpen
  \bibfield  {author} {\bibinfo {author} {\bibfnamefont {I.}~\bibnamefont
  {Pi\ifmmode~\check{z}\else \v{z}\fi{}orn}}, \bibinfo {author} {\bibfnamefont
  {L.}~\bibnamefont {Wang}}, \ and\ \bibinfo {author} {\bibfnamefont
  {F.}~\bibnamefont {Verstraete}},\ }\href {\doibase
  10.1103/PhysRevA.83.052321} {\bibfield  {journal} {\bibinfo  {journal} {Phys.
  Rev. A}\ }\textbf {\bibinfo {volume} {83}},\ \bibinfo {pages} {052321}
  (\bibinfo {year} {2011})},\ \Eprint {http://arxiv.org/abs/arXiv:1103.2343}
  {arXiv:1103.2343} \BibitemShut {NoStop}%
\bibitem [{\citenamefont {Lubasch}\ \emph
  {et~al.}(2014{\natexlab{a}})\citenamefont {Lubasch}, \citenamefont {Cirac},\
  and\ \citenamefont {Ba\~nuls}}]{ref:LCB2014-Cluster2}%
  \BibitemOpen
  \bibfield  {author} {\bibinfo {author} {\bibfnamefont {M.}~\bibnamefont
  {Lubasch}}, \bibinfo {author} {\bibfnamefont {J.~I.}\ \bibnamefont {Cirac}},
  \ and\ \bibinfo {author} {\bibfnamefont {M.-C.}\ \bibnamefont {Ba\~nuls}},\
  }\href {\doibase 10.1103/PhysRevB.90.064425} {\bibfield  {journal} {\bibinfo
  {journal} {Phys. Rev. B}\ }\textbf {\bibinfo {volume} {90}},\ \bibinfo
  {pages} {064425} (\bibinfo {year} {2014}{\natexlab{a}})},\ \Eprint
  {http://arxiv.org/abs/arXiv:1405.3259} {arXiv:1405.3259} \BibitemShut
  {NoStop}%
\bibitem [{\citenamefont {Jiang}\ \emph {et~al.}(2008)\citenamefont {Jiang},
  \citenamefont {Weng},\ and\ \citenamefont
  {Xiang}}]{ref:JWX2008-Simple-Update}%
  \BibitemOpen
  \bibfield  {author} {\bibinfo {author} {\bibfnamefont {H.~C.}\ \bibnamefont
  {Jiang}}, \bibinfo {author} {\bibfnamefont {Z.~Y.}\ \bibnamefont {Weng}}, \
  and\ \bibinfo {author} {\bibfnamefont {T.}~\bibnamefont {Xiang}},\ }\href
  {\doibase 10.1103/PhysRevLett.101.090603} {\bibfield  {journal} {\bibinfo
  {journal} {Phys. Rev. Lett.}\ }\textbf {\bibinfo {volume} {101}},\ \bibinfo
  {pages} {090603} (\bibinfo {year} {2008})},\ \Eprint
  {http://arxiv.org/abs/arXiv:0806.3719} {arXiv:0806.3719} \BibitemShut
  {NoStop}%
\bibitem [{\citenamefont {Lubasch}\ \emph
  {et~al.}(2014{\natexlab{b}})\citenamefont {Lubasch}, \citenamefont {Cirac},\
  and\ \citenamefont {Ba\~nuls}}]{ref:LCB2014-Cluster1}%
  \BibitemOpen
  \bibfield  {author} {\bibinfo {author} {\bibfnamefont {M.}~\bibnamefont
  {Lubasch}}, \bibinfo {author} {\bibfnamefont {J.~I.}\ \bibnamefont {Cirac}},
  \ and\ \bibinfo {author} {\bibfnamefont {M.-C.}\ \bibnamefont {Ba\~nuls}},\
  }\href {\doibase 10.1088/1367-2630/16/3/033014} {\bibfield  {journal}
  {\bibinfo  {journal} {New Journal of Physics}\ }\textbf {\bibinfo {volume}
  {16}},\ \bibinfo {pages} {033014} (\bibinfo {year} {2014}{\natexlab{b}})},\
  \Eprint {http://arxiv.org/abs/arXiv:1311.6696} {arXiv:1311.6696} \BibitemShut
  {NoStop}%
\bibitem [{\citenamefont {Hastings}(2004)}]{ref:Hastings2004-EXP}%
  \BibitemOpen
  \bibfield  {author} {\bibinfo {author} {\bibfnamefont {M.~B.}\ \bibnamefont
  {Hastings}},\ }\href {\doibase 10.1103/PhysRevB.69.104431} {\bibfield
  {journal} {\bibinfo  {journal} {Phys. Rev. B}\ }\textbf {\bibinfo {volume}
  {69}},\ \bibinfo {pages} {104431} (\bibinfo {year} {2004})},\ \Eprint
  {http://arxiv.org/abs/arXiv:cond-mat/0305505} {arXiv:cond-mat/0305505}
  \BibitemShut {NoStop}%
\bibitem [{\citenamefont {Hastings}\ and\ \citenamefont
  {Koma}(2006)}]{ref:HK2006-EXP}%
  \BibitemOpen
  \bibfield  {author} {\bibinfo {author} {\bibfnamefont {M.~B.}\ \bibnamefont
  {Hastings}}\ and\ \bibinfo {author} {\bibfnamefont {T.}~\bibnamefont
  {Koma}},\ }\href {\doibase 10.1007/s00220-006-0030-4} {\bibfield  {journal}
  {\bibinfo  {journal} {Communications in Mathematical Physics}\ }\textbf
  {\bibinfo {volume} {265}},\ \bibinfo {pages} {781} (\bibinfo {year}
  {2006})},\ \Eprint {http://arxiv.org/abs/arXiv:math-ph/0507008}
  {arXiv:math-ph/0507008} \BibitemShut {NoStop}%
\bibitem [{\citenamefont {Nachtergaele}\ and\ \citenamefont
  {Sims}(2006)}]{ref:Nachtergaele2006-LR}%
  \BibitemOpen
  \bibfield  {author} {\bibinfo {author} {\bibfnamefont {B.}~\bibnamefont
  {Nachtergaele}}\ and\ \bibinfo {author} {\bibfnamefont {R.}~\bibnamefont
  {Sims}},\ }\href {\doibase 10.1007/s00220-006-1556-1} {\bibfield  {journal}
  {\bibinfo  {journal} {Communications in Mathematical Physics}\ }\textbf
  {\bibinfo {volume} {265}},\ \bibinfo {pages} {119} (\bibinfo {year}
  {2006})},\ \Eprint {http://arxiv.org/abs/arXiv:math-ph/0506030}
  {arXiv:math-ph/0506030} \BibitemShut {NoStop}%
\bibitem [{\citenamefont {Kuwahara}\ \emph {et~al.}(2015)\citenamefont
  {Kuwahara}, \citenamefont {Arad}, \citenamefont {Amico},\ and\ \citenamefont
  {Vedral}}]{ref:KAAV2015-LR}%
  \BibitemOpen
  \bibfield  {author} {\bibinfo {author} {\bibfnamefont {T.}~\bibnamefont
  {Kuwahara}}, \bibinfo {author} {\bibfnamefont {I.}~\bibnamefont {Arad}},
  \bibinfo {author} {\bibfnamefont {L.}~\bibnamefont {Amico}}, \ and\ \bibinfo
  {author} {\bibfnamefont {V.}~\bibnamefont {Vedral}},\ }\href
  {http://arxiv.org/abs/1502.05330} {\bibfield  {journal} {\bibinfo  {journal}
  {arXiv preprint arXiv:1502.05330}\ } (\bibinfo {year} {2015})},\ \Eprint
  {http://arxiv.org/abs/arXiv:1502.05330} {arXiv:1502.05330} \BibitemShut
  {NoStop}%
\bibitem [{\citenamefont {Perez-Garcia}\ \emph {et~al.}(2008)\citenamefont
  {Perez-Garcia}, \citenamefont {Verstraete}, \citenamefont {Wolf},\ and\
  \citenamefont {Cirac}}]{ref:PVWC2008-UniquePEPS}%
  \BibitemOpen
  \bibfield  {author} {\bibinfo {author} {\bibfnamefont {D.}~\bibnamefont
  {Perez-Garcia}}, \bibinfo {author} {\bibfnamefont {F.}~\bibnamefont
  {Verstraete}}, \bibinfo {author} {\bibfnamefont {M.~M.}\ \bibnamefont
  {Wolf}}, \ and\ \bibinfo {author} {\bibfnamefont {J.~I.}\ \bibnamefont
  {Cirac}},\ }\href {http://dl.acm.org/citation.cfm?id=2016976.2016982}
  {\bibfield  {journal} {\bibinfo  {journal} {Quantum Info. Comput.}\ }\textbf
  {\bibinfo {volume} {8}},\ \bibinfo {pages} {650} (\bibinfo {year} {2008})},\
  \Eprint {http://arxiv.org/abs/arXiv:0707.2260} {arXiv:0707.2260} \BibitemShut
  {NoStop}%
\bibitem [{\citenamefont {Michalakis}\ and\ \citenamefont
  {Zwolak}(2013)}]{ref:MZ2013-LTQO1}%
  \BibitemOpen
  \bibfield  {author} {\bibinfo {author} {\bibfnamefont {S.}~\bibnamefont
  {Michalakis}}\ and\ \bibinfo {author} {\bibfnamefont {J.~P.}\ \bibnamefont
  {Zwolak}},\ }\href {\doibase 10.1007/s00220-013-1762-6} {\bibfield  {journal}
  {\bibinfo  {journal} {Communications in Mathematical Physics}\ }\textbf
  {\bibinfo {volume} {322}},\ \bibinfo {pages} {277} (\bibinfo {year}
  {2013})},\ \Eprint {http://arxiv.org/abs/arXiv:1109.1588} {arXiv:1109.1588}
  \BibitemShut {NoStop}%
\bibitem [{\citenamefont {Cirac}\ \emph {et~al.}(2013)\citenamefont {Cirac},
  \citenamefont {Michalakis}, \citenamefont {P\'erez-Garc\'{\i}a},\ and\
  \citenamefont {Schuch}}]{ref:CMPS2013-LTQO2}%
  \BibitemOpen
  \bibfield  {author} {\bibinfo {author} {\bibfnamefont {J.~I.}\ \bibnamefont
  {Cirac}}, \bibinfo {author} {\bibfnamefont {S.}~\bibnamefont {Michalakis}},
  \bibinfo {author} {\bibfnamefont {D.}~\bibnamefont {P\'erez-Garc\'{\i}a}}, \
  and\ \bibinfo {author} {\bibfnamefont {N.}~\bibnamefont {Schuch}},\ }\href
  {\doibase 10.1103/PhysRevB.88.115108} {\bibfield  {journal} {\bibinfo
  {journal} {Phys. Rev. B}\ }\textbf {\bibinfo {volume} {88}},\ \bibinfo
  {pages} {115108} (\bibinfo {year} {2013})}\BibitemShut {NoStop}%
\bibitem [{\citenamefont {Hastings}(2006)}]{ref:Hastings2006-GappedH}%
  \BibitemOpen
  \bibfield  {author} {\bibinfo {author} {\bibfnamefont {M.~B.}\ \bibnamefont
  {Hastings}},\ }\href {\doibase 10.1103/PhysRevB.73.085115} {\bibfield
  {journal} {\bibinfo  {journal} {Phys. Rev. B}\ }\textbf {\bibinfo {volume}
  {73}},\ \bibinfo {pages} {085115} (\bibinfo {year} {2006})},\ \Eprint
  {http://arxiv.org/abs/arXiv:cond-mat/0508554} {arXiv:cond-mat/0508554}
  \BibitemShut {NoStop}%
\bibitem [{\citenamefont {Affleck}\ \emph {et~al.}(1987)\citenamefont
  {Affleck}, \citenamefont {Kennedy}, \citenamefont {Lieb},\ and\ \citenamefont
  {Tasaki}}]{ref:AKLT1987}%
  \BibitemOpen
  \bibfield  {author} {\bibinfo {author} {\bibfnamefont {I.}~\bibnamefont
  {Affleck}}, \bibinfo {author} {\bibfnamefont {T.}~\bibnamefont {Kennedy}},
  \bibinfo {author} {\bibfnamefont {E.~H.}\ \bibnamefont {Lieb}}, \ and\
  \bibinfo {author} {\bibfnamefont {H.}~\bibnamefont {Tasaki}},\ }\href
  {\doibase 10.1103/PhysRevLett.59.799} {\bibfield  {journal} {\bibinfo
  {journal} {Phys. Rev. Lett.}\ }\textbf {\bibinfo {volume} {59}},\ \bibinfo
  {pages} {799} (\bibinfo {year} {1987})}\BibitemShut {NoStop}%
\bibitem [{\citenamefont {Affleck}\ \emph {et~al.}(1988)\citenamefont
  {Affleck}, \citenamefont {Kennedy}, \citenamefont {Lieb},\ and\ \citenamefont
  {Tasaki}}]{ref:AKLT1988}%
  \BibitemOpen
  \bibfield  {author} {\bibinfo {author} {\bibfnamefont {I.}~\bibnamefont
  {Affleck}}, \bibinfo {author} {\bibfnamefont {T.}~\bibnamefont {Kennedy}},
  \bibinfo {author} {\bibfnamefont {E.~H.}\ \bibnamefont {Lieb}}, \ and\
  \bibinfo {author} {\bibfnamefont {H.}~\bibnamefont {Tasaki}},\ }\href
  {http://projecteuclid.org/euclid.cmp/1104161001} {\bibfield  {journal}
  {\bibinfo  {journal} {Comm. Math. Phys.}\ }\textbf {\bibinfo {volume}
  {115}},\ \bibinfo {pages} {477} (\bibinfo {year} {1988})}\BibitemShut
  {NoStop}%
\bibitem [{\citenamefont {Knabe}(1988)}]{ref:Knabe1988-AKLTgap}%
  \BibitemOpen
  \bibfield  {author} {\bibinfo {author} {\bibfnamefont {S.}~\bibnamefont
  {Knabe}},\ }\href {\doibase 10.1007/BF01019721} {\bibfield  {journal}
  {\bibinfo  {journal} {Journal of Statistical Physics}\ }\textbf {\bibinfo
  {volume} {52}},\ \bibinfo {pages} {627} (\bibinfo {year} {1988})}\BibitemShut
  {NoStop}%
\bibitem [{\citenamefont {Yamashita}\ \emph {et~al.}(2012)\citenamefont
  {Yamashita}, \citenamefont {Fujisawa}, \citenamefont {Fukuda}, \citenamefont
  {Kobayashi}, \citenamefont {Nakta},\ and\ \citenamefont
  {Nakata}}]{ref:SDPAbook}%
  \BibitemOpen
  \bibfield  {author} {\bibinfo {author} {\bibfnamefont {M.}~\bibnamefont
  {Yamashita}}, \bibinfo {author} {\bibfnamefont {K.}~\bibnamefont {Fujisawa}},
  \bibinfo {author} {\bibfnamefont {M.}~\bibnamefont {Fukuda}}, \bibinfo
  {author} {\bibfnamefont {K.}~\bibnamefont {Kobayashi}}, \bibinfo {author}
  {\bibfnamefont {K.}~\bibnamefont {Nakta}}, \ and\ \bibinfo {author}
  {\bibfnamefont {M.}~\bibnamefont {Nakata}},\ }in\ \href@noop {} {\emph
  {\bibinfo {booktitle} {{In Handbook on Semidefinite, Cone and Polynomial
  Optimization: Theory, Algorithms, Software and Applications}}}},\ \bibinfo
  {editor} {edited by\ \bibinfo {editor} {\bibfnamefont {M.~F.}\ \bibnamefont
  {Anjos}}\ and\ \bibinfo {editor} {\bibfnamefont {J.~B.}\ \bibnamefont
  {Lasserre}}}\ (\bibinfo  {publisher} {Springer},\ \bibinfo {address} {NY,
  USA},\ \bibinfo {year} {2012})\ Chap.~\bibinfo {chapter} {24}, pp.\ \bibinfo
  {pages} {687--714}\BibitemShut {NoStop}%
\bibitem [{\citenamefont {Yamashita}\ \emph {et~al.}(2010)\citenamefont
  {Yamashita}, \citenamefont {Fujisawa}, \citenamefont {Nakata}, \citenamefont
  {Nakata}, \citenamefont {Fukuda}, \citenamefont {Kobayashi},\ and\
  \citenamefont {Goto}}]{ref:SDPArep}%
  \BibitemOpen
  \bibfield  {author} {\bibinfo {author} {\bibfnamefont {M.}~\bibnamefont
  {Yamashita}}, \bibinfo {author} {\bibfnamefont {K.}~\bibnamefont {Fujisawa}},
  \bibinfo {author} {\bibfnamefont {K.}~\bibnamefont {Nakata}}, \bibinfo
  {author} {\bibfnamefont {M.}~\bibnamefont {Nakata}}, \bibinfo {author}
  {\bibfnamefont {M.}~\bibnamefont {Fukuda}}, \bibinfo {author} {\bibfnamefont
  {K.}~\bibnamefont {Kobayashi}}, \ and\ \bibinfo {author} {\bibfnamefont
  {K.}~\bibnamefont {Goto}},\ }\href@noop {} {\emph {\bibinfo {title} {{A
  high-performance software package for semidefinite programs: SDPA 7}}}},\
  \bibinfo {type} {Tech. Rep.}\ \bibinfo {number} {B-460}\ (\bibinfo
  {institution} {Dept. of Mathematical and Computing Science, Tokyo Institute
  of Technology},\ \bibinfo {address} {Tokyo, Japan},\ \bibinfo {year} {2010})\
  \bibinfo {note} {\url{http://sdpa.sourceforge.net/index.html}}\BibitemShut
  {NoStop}%
\bibitem [{Note1()}]{Note1}%
  \BibitemOpen
  \bibinfo {note} {Strictly speaking, one can have a frustrated system in which
  the groundstate is a common eigenstate -- for example, a frustrated classical
  system. However, in some sense these are not `genuine' frustrated systems,
  since they can be trivially transformed into a frustration-free system by by
  locally subtracting an identity term from every local term in the
  Hamiltonian. This would make the system frustration-free and only shift its
  spectrum by a constant.}\BibitemShut {Stop}%
\bibitem [{\citenamefont {Johnson}\ and\ \citenamefont
  {Lindenstrauss}(1984)}]{ref:JL1984}%
  \BibitemOpen
  \bibfield  {author} {\bibinfo {author} {\bibfnamefont {W.~B.}\ \bibnamefont
  {Johnson}}\ and\ \bibinfo {author} {\bibfnamefont {J.}~\bibnamefont
  {Lindenstrauss}},\ }\href@noop {} {\bibfield  {journal} {\bibinfo  {journal}
  {Contemporary mathematics}\ }\textbf {\bibinfo {volume} {26}},\ \bibinfo
  {pages} {1} (\bibinfo {year} {1984})}\BibitemShut {NoStop}%
\bibitem [{\citenamefont {Evenbly}\ and\ \citenamefont
  {Vidal}(2010)}]{ref:EV2010-Kagome}%
  \BibitemOpen
  \bibfield  {author} {\bibinfo {author} {\bibfnamefont {G.}~\bibnamefont
  {Evenbly}}\ and\ \bibinfo {author} {\bibfnamefont {G.}~\bibnamefont
  {Vidal}},\ }\href {\doibase 10.1103/PhysRevLett.104.187203} {\bibfield
  {journal} {\bibinfo  {journal} {Phys. Rev. Lett.}\ }\textbf {\bibinfo
  {volume} {104}},\ \bibinfo {pages} {187203} (\bibinfo {year} {2010})},\
  \Eprint {http://arxiv.org/abs/arXiv:0904.3383} {arXiv:0904.3383} \BibitemShut
  {NoStop}%
\bibitem [{\citenamefont {Yan}\ \emph {et~al.}(2011)\citenamefont {Yan},
  \citenamefont {Huse},\ and\ \citenamefont {White}}]{ref:YHW2011-Kagome}%
  \BibitemOpen
  \bibfield  {author} {\bibinfo {author} {\bibfnamefont {S.}~\bibnamefont
  {Yan}}, \bibinfo {author} {\bibfnamefont {D.~A.}\ \bibnamefont {Huse}}, \
  and\ \bibinfo {author} {\bibfnamefont {S.~R.}\ \bibnamefont {White}},\ }\href
  {\doibase 10.1126/science.1201080} {\bibfield  {journal} {\bibinfo  {journal}
  {Science}\ }\textbf {\bibinfo {volume} {332}},\ \bibinfo {pages} {1173}
  (\bibinfo {year} {2011})},\ \Eprint {http://arxiv.org/abs/arXiv:1011.6114}
  {arXiv:1011.6114} \BibitemShut {NoStop}%
\bibitem [{\citenamefont {Depenbrock}\ \emph {et~al.}(2012)\citenamefont
  {Depenbrock}, \citenamefont {McCulloch},\ and\ \citenamefont
  {Schollw\"ock}}]{ref:DMS2012-Kagome}%
  \BibitemOpen
  \bibfield  {author} {\bibinfo {author} {\bibfnamefont {S.}~\bibnamefont
  {Depenbrock}}, \bibinfo {author} {\bibfnamefont {I.~P.}\ \bibnamefont
  {McCulloch}}, \ and\ \bibinfo {author} {\bibfnamefont {U.}~\bibnamefont
  {Schollw\"ock}},\ }\href {\doibase 10.1103/PhysRevLett.109.067201} {\bibfield
   {journal} {\bibinfo  {journal} {Phys. Rev. Lett.}\ }\textbf {\bibinfo
  {volume} {109}},\ \bibinfo {pages} {067201} (\bibinfo {year} {2012})},\
  \Eprint {http://arxiv.org/abs/arXiv:1205.4858} {arXiv:1205.4858} \BibitemShut
  {NoStop}%
\bibitem [{\citenamefont {Iqbal}\ \emph {et~al.}(2013)\citenamefont {Iqbal},
  \citenamefont {Becca}, \citenamefont {Sorella},\ and\ \citenamefont
  {Poilblanc}}]{ref:IBSP2013-Kagome}%
  \BibitemOpen
  \bibfield  {author} {\bibinfo {author} {\bibfnamefont {Y.}~\bibnamefont
  {Iqbal}}, \bibinfo {author} {\bibfnamefont {F.}~\bibnamefont {Becca}},
  \bibinfo {author} {\bibfnamefont {S.}~\bibnamefont {Sorella}}, \ and\
  \bibinfo {author} {\bibfnamefont {D.}~\bibnamefont {Poilblanc}},\ }\href
  {\doibase 10.1103/PhysRevB.87.060405} {\bibfield  {journal} {\bibinfo
  {journal} {Phys. Rev. B}\ }\textbf {\bibinfo {volume} {87}},\ \bibinfo
  {pages} {060405} (\bibinfo {year} {2013})},\ \Eprint
  {http://arxiv.org/abs/arXiv:1209.1858} {arXiv:1209.1858} \BibitemShut
  {NoStop}%
\bibitem [{\citenamefont {Capponi}\ \emph {et~al.}(2013)\citenamefont
  {Capponi}, \citenamefont {Chandra}, \citenamefont {Auerbach},\ and\
  \citenamefont {Weinstein}}]{ref:CCAW2013-Kagome}%
  \BibitemOpen
  \bibfield  {author} {\bibinfo {author} {\bibfnamefont {S.}~\bibnamefont
  {Capponi}}, \bibinfo {author} {\bibfnamefont {V.~R.}\ \bibnamefont
  {Chandra}}, \bibinfo {author} {\bibfnamefont {A.}~\bibnamefont {Auerbach}}, \
  and\ \bibinfo {author} {\bibfnamefont {M.}~\bibnamefont {Weinstein}},\ }\href
  {\doibase 10.1103/PhysRevB.87.161118} {\bibfield  {journal} {\bibinfo
  {journal} {Phys. Rev. B}\ }\textbf {\bibinfo {volume} {87}},\ \bibinfo
  {pages} {161118} (\bibinfo {year} {2013})},\ \Eprint
  {http://arxiv.org/abs/arXiv:1210.5519} {arXiv:1210.5519} \BibitemShut
  {NoStop}%
\bibitem [{\citenamefont {Aharonov}\ \emph {et~al.}(2011)\citenamefont
  {Aharonov}, \citenamefont {Arad}, \citenamefont {Vazirani},\ and\
  \citenamefont {Landau}}]{ref:AAVZ2011-DL}%
  \BibitemOpen
  \bibfield  {author} {\bibinfo {author} {\bibfnamefont {D.}~\bibnamefont
  {Aharonov}}, \bibinfo {author} {\bibfnamefont {I.}~\bibnamefont {Arad}},
  \bibinfo {author} {\bibfnamefont {U.}~\bibnamefont {Vazirani}}, \ and\
  \bibinfo {author} {\bibfnamefont {Z.}~\bibnamefont {Landau}},\ }\href
  {\doibase 10.1088/1367-2630/13/11/113043} {\bibfield  {journal} {\bibinfo
  {journal} {New Journal of Physics}\ }\textbf {\bibinfo {volume} {13}},\
  \bibinfo {pages} {113043} (\bibinfo {year} {2011})},\ \Eprint
  {http://arxiv.org/abs/arXiv:1011.3445} {arXiv:1011.3445} \BibitemShut
  {NoStop}%
\bibitem [{\citenamefont {{Anshu}}\ \emph {et~al.}(2016)\citenamefont
  {{Anshu}}, \citenamefont {{Arad}},\ and\ \citenamefont
  {{Vidick}}}]{ref:AAV2016-DLD}%
  \BibitemOpen
  \bibfield  {author} {\bibinfo {author} {\bibfnamefont {A.}~\bibnamefont
  {{Anshu}}}, \bibinfo {author} {\bibfnamefont {I.}~\bibnamefont {{Arad}}}, \
  and\ \bibinfo {author} {\bibfnamefont {T.}~\bibnamefont {{Vidick}}},\
  }\href@noop {} {\bibfield  {journal} {\bibinfo  {journal} {ArXiv e-prints}\ }
  (\bibinfo {year} {2016})},\ \Eprint {http://arxiv.org/abs/arXiv:1602.01210}
  {arXiv:arXiv:1602.01210 [quant-ph]} \BibitemShut {NoStop}%
\end{thebibliography}%

\end{document}